\documentclass[11pt, letterpaper]{article}
\usepackage[utf8]{inputenc}
\usepackage{graphicx}
\usepackage[margin=1in]{geometry}
\usepackage{amsmath,amssymb,amsthm,verbatim}
\usepackage[ruled,vlined]{algorithm2e}
\usepackage{algorithmic}
\usepackage{hyperref}
\usepackage{xcolor}
\usepackage{footmisc}
\usepackage{array}
\usepackage{csquotes}
\hypersetup{ colorlinks = true, allcolors=blue, linkcolor=red, urlcolor=blue}
\usepackage{framed}

\theoremstyle{definition}
\newtheorem{definition}{Definition}[section]
\usepackage{thm-restate}
\newtheorem{theorem}{Theorem}
\newtheorem{lemmma}[theorem]{Lemma}

\newtheorem{corollary}[theorem]{Corollary}

\usepackage[square,sort,comma,numbers]{natbib}

\renewcommand{\arraystretch}{2}
\renewcommand{\beta}{{a}}
\renewcommand{\gamma}{{b}}
\newcommand{\ALG}{\textrm{ALG}}
\newcommand{\OPT}{\textrm{OPT}}

\newtheorem*{remark}{Remark}

\title{Prophet Inequality: Order selection beats random order}

\author{
Archit Bubna \\ Indian Institute of Technology, Delhi  \\ \href{mailto:architbubna12@gmail.com}{architbubna12@gmail.com}
\and
Ashish Chiplunkar\\ Indian Institute of Technology, Delhi  \\ \href{mailto:ashishc@iitd.ac.in}{ashishc@iitd.ac.in}
}
\date{}

\begin{document}

\maketitle

\begin{abstract}

In the prophet inequality problem, a gambler faces a sequence of items arriving online with values drawn independently from known distributions. On seeing an item, the gambler must choose whether to accept its value as her reward and quit the game, or reject it and continue. The gambler's aim is to maximize her expected reward relative to the expected maximum of the values of all items.
Since the seventies, a tight bound of $\frac{1}{2}$ has been known for this competitive ratio in the setting where the items arrive in an adversarial order \citep*{Krengel1,Krengel2}. However, the optimum ratio still remains unknown in the order selection setting, where the gambler selects the arrival order, as well as in \textit{prophet secretary}, where the items arrive in a random order. Moreover, it is not even known whether a separation exists between the two settings.

In this paper, we show that the power of order selection allows the gambler to guarantee a strictly better competitive ratio than if the items arrive randomly. 
For the order selection setting, we identify an instance for which Peng and Tang’s \citep*{PT} state-of-the-art algorithm performs no better than their claimed competitive ratio of (approximately) $0.7251$, thus illustrating the need for an improved approach. We therefore extend their design and provide a more general algorithm design framework, using which we show that their ratio can be beaten, by designing a 0.7258-competitive algorithm.
For the random order setting, we improve upon Correa, Saona and Ziliotto's \citep{Correa} $\sqrt{3}-1\approx$ 0.732-hardness result to show a hardness of 0.7254 for general algorithms - even in the setting where the gambler knows the arrival order beforehand, thus establishing a separation between the order selection and random order settings. 
\end{abstract}


\section{Introduction}\label{intro}
The prophet inequality is a cornerstone result in optimal stopping theory which concerns a game where a gambler faces a sequence of items that arrive online, with value drawn independently from distributions known to the gambler. Once an item arrives, the gambler can see its realized value and she must choose to either accept it as her reward and end the game, or reject it irrevocably. The gambler's goal is to maximize her reward, and compete against the expected reward accumulated by a prophet, who knows the value of each item beforehand, and hence only accepts the item with the maximum value. The prophet inequality due to \citet{Krengel1,Krengel2} asserts that, when the arrival order is adversarial, a $\frac{1}{2}$-competitive algorithm exists, that is, there exists an algorithm that enables the gambler to win a reward whose expectation is at least $\frac{1}{2}$ times the expected reward of the prophet. Moreover, no algorithm can guarantee a competitive ratio greater than $\frac{1}{2}$.

The prophet inequality problem is closely connected to \textit{posted price mechanisms} or PPMs. In a PPM, there is a seller who wants to sell an item, and a sequence of $N$ buyers who approach the seller one-at-a-time. Upon the arrival of a buyer, the seller offers her a price. The buyer may choose to accept or reject the offer, based on her valuation of the item. The first buyer to accept the offer gets the item and pays the seller the price that was offered to her. \citet{PPM} and \citet{Chawla} demonstrated that the problem of designing posted price mechanisms can be reduced to the prophet inequality problem. Later, \citet{CorreaPPM} showed a reduction in the opposite direction, demonstrating that the two problems are equivalent. \citet{Lucier} performs an extensive survey on the economic implications of prophet inequalities.

The economic relevance of the prophet inequality problem has led to an increased interest in the study of the problem and its variants in recent years. The study of many of these variants is aimed at beating the $\frac{1}{2}$ barrier in more relaxed settings. We now describe some of the most extensively studied variants of the prophet inequality problem, starting from the least relaxed setting to the most relaxed setting.
\begin{itemize}
    \item \textbf{Random Order Setting (or Prophet Secretary):} This variant, where the items arrive in a uniformly randomly chosen order, was first studied by \citet{Esfandiari}, who gave a $1-\frac{1}{e}\approx0.632$-competitive algorithm. Later, \citet{Correa0.632} and \citet{Ehsani} used different methods to achieve the same $1-\frac{1}{e}$ ratio. The $1-\frac{1}{e}$ barrier was first beaten by \citet{Azar}. Later, \citet{Correa} provided a $0.669$-competitive algorithm, which still stands as the best known bound for the random order setting. They also showed a hardness of $\sqrt{3}-1\approx 0.732$, thus showing for the first time that the prophet inequality problem is strictly harder in the random order setting than in the IID setting.
    \item \textbf{Order Selection Setting:} In this setting, the gambler is allowed to select the order in which the items arrive. This setting is no harder than the random order setting. For a long time, the best-known bound for the competitive-ratio in the order selection setting was $1-\frac{1}{e}$, as shown by \citet{Chawla}. This bound was believed to be tight because the $1-\frac{1}{e}$ ratio is optimal if we compare the algorithm against the \textit{ex-ante relaxation} objective. Later, the $1-\frac{1}{e}$ bound underwent a sequence of improvements \citep{Azar, Beyhaghi, Correa}. Very recently, \citet{PT} achieved a major improvement by using a novel technique, which they call continuous arrival time design, to construct an (approximately) $0.7251$-competitive\footnote{Peng and Tang's competitive ratio of $\Gamma_{PT}\approx 0.7251$ is given by $\Gamma_{PT}=\frac{\ln \alpha + 1}{\ln \alpha + 1 - \alpha}$, where $\alpha \approx 0.2109$ is the unique solution to $\int_\alpha^1 \frac{\ln\alpha+1}{(\ln\alpha+1)(-x\ln x+x)-\alpha}dx+\frac{1}{\alpha}=0$ on $(0,1)$.\label{PT footnote} } algorithm for this setting. 
    \item \textbf{IID Setting:} In this setting, all the value distributions are identical. Observe that this setting is no harder than the order selection setting in the worst case. This is because the IID setting can be seen as a special case of the order selection setting, since the power of order selection is useless in the case of identical distributions.  \citet{Hill} initiated the study of this setting and gave a $1-\frac{1}{e}\approx 0.632$-competitive algorithm. They also showed that no algorithm can obtain a competitive ratio greater than (approximately) $0.745$.\footnote{Hill and Kertz's IID bound of $\Gamma \approx 0.745$ is the unique solution to $\int_0^1\frac{1}{y(1-\ln y)+1/\Gamma-1}dy=1$.} Later, \citet{Abolhassani} improved the $0.632$ competitive-ratio to $0.738$. The problem was finally closed by \citet{Correa2}, who designed an algorithm that matched Hill and Kertz's hardness bound of 0.745. It is noteworthy that the work of \citet{Allaart} obtains the same competitive ratio for a closely related problem in which independent samples from a fixed probability distribution are generated by a Poisson process, and the algorithm is required to pick one of them irrevocably within a finite time horizon.
\end{itemize}

The arrival order of the buyers is of great importance in PPM design. What welfare guarantees can we provide if the buyers approach the seller in a random order? Can we do any better if the seller gets to select the arrival order of the buyers? The latter has been an important open problem, as sequential posted price mechanisms \citep{Chawla} allow the seller to select the order in which the buyers arrive, making it natural to study the benefit of order selection. The random order and order selection settings of the prophet inequality problem directly correspond to the above mentioned scenarios in PPM design. Bounds obtained for these variants of the prophet inequality problem can be directly used as bounds on the total welfare achieved through PPMs relative to offline welfare maximizing auctions.

\subsection{Our Results and Techniques}
In this paper, we establish a separation between the order selection setting and the random order setting by constructing a 0.7258-competitive algorithm for the former and showing a hardness of 0.7254 for the latter. 

For the order selection setting, \citet{PT} provided an (approximately) 0.7251-competitive\footref{PT footnote} algorithm that uses a continuous arrival time design, where each item $i$ is assumed to arrive at some timestamp $t_i \in [0,1]$. The algorithm draws the arrival time $t_i$ of each item independently from carefully constructed arrival time distributions, and subjects them to a common time-dependent threshold function. We extend their design and provide a more general framework to design algorithms that use independent arrival times. This framework allows the algorithm to use a different time-dependent threshold function for each item. Using this framework, we show that Peng and Tang's competitive ratio can be beaten. We prove,
\begin{restatable}{theorem}{orderselection} \label{Order Selection}
There exists an algorithm that guarantees a competitive ratio of at least $\Gamma^*=0.7258$ for all instances of the order selection prophet inequality problem.
\end{restatable}
Note that our main contribution in the order selection setting is not the numerical improvement in the ratio, but rather it is the demonstration of the fact that Peng and Tang's ratio can be beaten by relaxing the constraints of their algorithm and using a more general approach. In order to motivate the need for this more general approach, we identify a set of distributions for which Peng and Tang's algorithm can perform no better than their claimed competitive ratio. This instance is composed of $N$ IID variables whose maximum is distributed uniformly over $[0,1]$, and another variable which is uniformly distributed over the interval $[\alpha,\alpha+\frac{1}{N}]$, where $\alpha\approx0.2109$.\footnote{$\alpha \approx 0.2109$ is the unique solution to $\int_\alpha^1 \frac{\ln\alpha+1}{(\ln\alpha+1)(-x\ln x+x)-\alpha}dx+\frac{1}{\alpha}=0$ on $(0,1)$. } Using straightforward but tedious calculations, it can be verified that as $N$ approaches $\infty$, the maximum competitive ratio that Peng and Tang's algorithm can achieve for the above instance approaches $\Gamma_{PT}\approx 0.7251$.\footref{PT footnote} 

Peng and Tang's algorithm satisfies a stronger condition than competitiveness - the reward $\ALG$ accumulated by the algorithm satisfies $\mathbb{P}[\ALG > x]\geq\Gamma\cdot\mathbb{P}[\max_i v_i > x] \textrm{\ \ \ for all } x>0$ (for $\Gamma=\Gamma_{PT}$ in their case). Here, $v_i$ denotes the value of the $i$-th item.
We call this condition $\Gamma-$\textit{approximate stochastic dominance} (ASD), and we shall refer to such an algorithm as a $\Gamma$-ASD algorithm. It is easy to see that this condition is no weaker than $\Gamma$-competitiveness, i.e. a $\Gamma$-ASD algorithm is $\Gamma$-competitive as well. We show a surprising result, \footnote{This result and its proof are analogous to the result of \citet{Lee}, which states that the existence of a \textit{$\Gamma$-competitive ex-ante prophet inequality} implies the existence of a \textit{$\Gamma$-selectable online contention resolution scheme}.}

\begin{restatable}{theorem}{LP} \label{LP}
If there exists a $\Gamma$-competitive algorithm for all finite support instances of the order selection prophet inequality, then there also exists a $\Gamma$-ASD algorithm for all finite support instances of the order selection prophet inequality.
\end{restatable} 
In fact, this result also applies to a wide range of arrival order settings, including the random order setting (prophet secretary) and the constrained order setting \citep*{constrained}, and hence we believe this result would be of independent interest. It is convenient to have this result, because designing and analysing a $\Gamma$-ASD algorithm could be much simpler than showing the $\Gamma$-competitiveness of an algorithm without the ASD property. The state-of-the-art results for both the order selection setting and random order setting \citep{Correa} use the ASD property for their competitive analysis. The competitive analysis of algorithms generated by our framework also relies on this ASD condition. The framework allows us to choose a set of identity-dependent threshold functions $\{\tau_i(t)\}$ for the items, and we show that we can find a set of thresholds for which the algorithm generated by the framework guarantees $\Gamma^*$-ASD for $\Gamma^*=0.7258>\Gamma_{PT}$. This opens up the possibility for further exploration in this direction, i.e. finding a set of threshold functions $\{\tau_i(t)\}$ for which the independent arrival time framework provides a significantly higher (possibly optimal) ASD guarantee. For the special case when all the threshold functions are identical, the algorithm generated by the framework is the same as Peng and Tang's algorithm, which achieves $\Gamma_{PT}$-ASD. 

The existing literature on the prophet secretary problem does not clearly define whether or not the algorithm is aware of the arrival order of the items beforehand. Both the possibilities are fairly natural, and we shall refer to them as \textit{order-aware} prophet secretary and \textit{order-unaware} prophet secretary. Our hardness result of 0.7254 for prophet secretary applies to the \textit{order-aware} setting (and hence also to the \textit{order-unaware} setting), thus separating both variants of the prophet secretary problem from the order selection setting.
\begin{restatable}{theorem}{randomorder} \label{Random Order}
There exists no algorithm that guarantees a competitive ratio greater than $0.7254$ for all instances of the \textit{order-aware} prophet secretary problem.
\end{restatable}
This result improves upon the 0.732-hardness result due to \citet{Correa}, who analysed a hard instance composed of $N$ two-point IID variables and one deterministic variable. We extend their construction by increasing the support size of the IID variables, allowing us to obtain this improved hardness result. Following our result, a more recent work by \citet{Giordano} used the same idea to show a hardness of 0.7235, however, their result only applies to the \textit{order-unaware} setting.

\subsection{Related Work}
\textbf{Order Selection Prophet Inequality:} 
\citet{Abolhassani} showed that a competitive ratio of 0.738 can be attained for instances where each type of distribution occurs $\Omega(\log n)$ times. \citet{Liu} showed that if the algorithm is allowed to remove a constant number of items, then it can attain a competitive ratio that is arbitrarily close to the IID bound of 0.745.

\noindent \textbf{Optimal Ordering:} This problem deals with maximizing the algorithm's reward relative to the optimal online algorithm's reward in the order selection setting. The problem of selecting the optimal arrival order was shown to be NP-hard by \citet{Agrawal}.
\citet{Chakraborty} designed a PTAS for the optimal ordering problem, which was improved to an EPTAS by \citet{Liu}.

\noindent \textbf{Matroid Prophet Inequalities:} In the matroid prophet inequality, the algorithm is allowed to pick a set of items, with the feasible sets of items being independent sets of a given matroid. For $k$-uniform matroids, \citet{Alaei} obtained an asymptotically optimal competitive ratio of $1-O(k^{-\frac{1}{2}})$. \citet{Kleinberg} gave a $\frac{1}{2}$-competitive algorithm for general matroids. Analogously, in the random order setting, \citet{Ehsani} gave a $\left(1-\frac{1}{e}\right)$-competitive algorithm for general matroid feasibility constraints.

\subsection{Organization of the Paper}
Sections \ref{Preliminaries} to \ref{2scheme} of this paper are dedicated to proving Theorem \ref{Order Selection}. More specifically, in Section \ref{General Framework}, we describe our general framework for constructing algorithms that use independent arrival times.  In Section \ref{2scheme}, we construct our 0.7258-competitive algorithm using this framework. Section \ref{LPproof} contains the proof for Theorem \ref{LP}. The proof for Theorem \ref{Random Order}  is contained in Section \ref{Random Order Section} of the paper, and can be read independently from the rest of the paper.

\section{Preliminaries: The Order Selection Setting}\label{Preliminaries}
An instance of the order selection prophet inequality problem is composed of $n>1$ items, along with their corresponding probability distributions $D_1, D_2, D_3,\ldots,D_n$ which are known to the algorithm. The values $\{v_i\}$ of the items are drawn independently from the distributions $\{D_i\}$. The algorithm first selects the order in which the $n$ items arrive. Once the $i$'th item arrives, the algorithm is shown its value $v_i$, and the algorithm must choose to either accept the item and stop, or reject the item irrevocably and move to the next. The aim is to maximize the expected value of the item accepted by the algorithm and compete against a prophet, who can see the future and only accepts the item with the maximum value. We call an algorithm $\Gamma$-competitive if it satisfies the following relation:
$$\mathbb{E}[\ALG]\geq\Gamma\cdot\mathbb{E}[\max_i v_i],$$
where $\ALG$ is a random variable denoting the value of the item accepted by the algorithm.

\subsection{Approximate Stochastic Dominance}\label{ASD}
We say that a non-negative random variable $X$ attains $\Gamma$-approximate stochastic dominance (or $\Gamma$-ASD) over another non-negative random variable $Y$ if the following relation is satisfied.
$$\mathbb{P}[X> x]\geq\Gamma\cdot\mathbb{P}[Y > x] \textrm{\ \ \ for all } x>0.$$
Using the fact that the relation $\int_0^\infty \mathbb{P}[Z > x]\cdot dx = \mathbb{E}[Z]$ holds for every non-negative random variable $Z$ with finite mean, and integrating the above inequality from $x=0$ to $\infty$ on both sides, we directly obtain that $X$ attaining $\Gamma$-ASD over $Y$  implies $\mathbb{E}[X]\geq\Gamma\cdot\mathbb{E}[Y]$. It follows from here that if an algorithm attains $\Gamma$-ASD over the prophet (i.e. $\max_i v_i$), then it is also $\Gamma$-competitive, i.e.
$$\mathbb{P}[\ALG > x]\geq\Gamma\cdot\mathbb{P}[\max_i v_i > x] \textrm{\ \ \ for all } x>0 \ \ \implies \ \  \mathbb{E}[\ALG]\geq\Gamma\cdot\mathbb{E}[\max_i v_i].$$

\subsection{Arrival Time Design and Notation}
We use a continuous arrival time design, similar to the one used by \citet{PT}, wherein we assume that each item $i$ arrives at a time $t_i \in [0,1]$. They define fixed time-dependent thresholds given by the function $\tau(t)$, where $\tau(t)$ is given as

$$\mathbb{P}[\max_{i}v_i>\tau(t)]=t.$$
An item $i$ arriving at time $t_i$ is accepted by their algorithm if the algorithm reaches the item and $v_i > \tau(t_i)$. We preserve the definition of this notation in this paper. We also borrow the following notations. For every $t\in [0,1]$ and $i \in [n]$, we define
$$p_i(t)\overset{\underset{\mathrm{def}}{}}{=}\mathbb{P}[v_i>\tau(t)] \textrm{\ \ \ \ \ and\ \ \ \ \ } q_i(t)\overset{\underset{\mathrm{def}}{}}{=}\mathbb{P}[\max_{j \neq i}v_j>\tau(t)].$$
Throughout Sections \ref{General Framework} and \ref{2scheme}, we assume that we are working with continuous distributions $\{D_i\}$, and hence it is safe to assume that $p_i(t)$ and $q_i(t)$ are non-decreasing continuous functions of time. To see how discrete distributions can be handled, the reader may refer to \citet{Correa}.

From the definitions, it is clear that $p_i(0)=q_i(0)=0$ for all $i \in [n]$. For the ease of presentation, we also assume that all the distributions are supported on some contiguous interval of real numbers, with non-zero probability densities throughout the interval. This allows us to assume that $p_i(1)=q_i(1)=1$ holds for all $i \in [n]$, and that $p_i(t)$ and $q_i(t)$ are strictly increasing on $[0,1]$.

\subsection{Peng and Tang's Independent Arrival Time Algorithm}\label{Peng and Tang}

In this section, we give a brief description of Peng and Tang's \citep{PT} independent arrival time algorithm, which forms a critical subroutine in our algorithm. We also develop some useful notation along the way. Note that we have slightly modified the presentation of the algorithm, considering our assumption that $p_i(1)=q_i(1)=1$ holds for all $i \in [n]$. For some $\Gamma \in (0,1)$:
\begin{itemize}
    \item The algorithm samples the arrival time $t_i$ of each item $i$ independently from carefully constructed distributions. For each $i$, the probability density of $t_i$ at $t\in[0,1)$ is given by a function $f_i(t,\Gamma)$. We let $t_i=1$ with probability $1-\int_0^1 f_i(t,\Gamma)\cdot dt$.
    \item The items are made to arrive in ascending order of their arrival times.
    \item The algorithm accepts the first item that satisfies $v_i>\tau(t_i)$, where $\tau(t)$ satisfies
    $$\mathbb{P}[\max_{i}v_i>\tau(t)]=t.$$
    \item An item appearing at $t=1$ is always rejected.
\end{itemize}
Before we define the arrival time distributions, we define an auxiliary function $g(t, \Gamma)$. 
$$g(t, \Gamma) \overset{\underset{\mathrm{def}}{}}{=} \Gamma \left( \sum_i(1-q_i(t))p_i(t)-t\right)+1.$$
We now define the arrival time distributions $\{f_i(t,\Gamma)\}$.

$$f_i(t, \Gamma) \overset{\underset{\mathrm{def}}{}}{=} \Gamma\frac{q_i'(t)}{g(t, \Gamma)}\exp\left ( -\Gamma \int_{0}^{t}\frac{q_i'(s)p_i(s)}{g(s, \Gamma)}ds\right ).$$ We use the above definitions for $f_i(t, \Gamma)$ and $g(t, \Gamma)$ throughout this paper.
\bigbreak
\noindent \textbf{Condition for the construction to be well defined:} The only condition required for this construction to be well defined is that 
$$\int_0^1f_i(t,\Gamma)\cdot dt \leq 1 \textrm{\  for all\  }{i \in [n]}.$$
Now, we state a result from \citet{PT}, which says that the  construction is always well defined if $\Gamma = \Gamma_{PT} \approx 0.7251$. 

\begin{theorem}\label{Condition for Peng and Tang to be valid}
For any set of distributions $\{D_i\}$ with $n$ items, the following inequality holds for all $i \in [n]$.
$$\int_0^1f_i(t,\Gamma_{PT})\cdot dt \leq 1 \textrm{\  for all\  }{i \in [n]}$$
where $\Gamma_{PT} = \frac{\ln\alpha+1}{\ln\alpha+1-\alpha} \approx 0.7251$ and $\alpha \approx 0.2109$ is the unique solution to the following equation on $(0,1)$.
$$\int_\alpha^1 \frac{\ln\alpha+1}{(\ln\alpha+1)(-x\ln x+x)-\alpha}dx+\frac{1}{\alpha}=0$$
\end{theorem}
We now state a result from \citet{PT} about the attainment of ASD (refer to Section \ref{ASD}) by the algorithm. Here, the random variable $\ALG$ denotes the value of the item accepted by the algorithm.

\begin{theorem}\label{ASD and CR of Peng and Tang}
If for some $\Gamma \in (0,1)$ the arrival time distributions $\{f_i(t, \Gamma)\}$ are well defined, then for all $t\in [0,1]$ the following inequality holds
$$\mathbb{P}[\ALG>\tau(t)] \geq \Gamma \cdot  \mathbb{P}[\max_{i}v_i>\tau(t)]$$ when the algorithm selects the arrival times of the items from the distributions $\{f_i(t, \Gamma)\} $ and uses the threshold function $\tau(t)$.
\end{theorem}

\section{A General Framework for Independent Arrival Time Algorithms}\label{General Framework}

In this section, we extend Peng and Tang's \citep{PT} algorithm and provide a general framework for algorithms that select the arrival time for each item independently. For this, we expand into the space of algorithms that use a (not necessarily) different strictly decreasing threshold function $\tau_i(t)$ for each item $i$. For a constant $\Gamma \in (0,1)$, the following procedure is followed by such an algorithm:
\begin{itemize}
    \item The algorithm samples the arrival time $t_i$ of each item $i$ independently from carefully constructed distributions. For each $i$, the probability density of $t_i$ at $t\in[0,1)$ is given by a function $\bar{f}_i(t,\Gamma)$. We let $t_i=1$ with probability $1-\int_0^1 \bar{f}_i(t,\Gamma)\cdot dt$.
    \item The items are made to arrive in ascending order of their arrival times.
    \item The algorithm accepts the first item that satisfies $v_i>\tau_i(t_i)$, where $\tau_i(t)$ is the threshold function for item $i$. \textbf{This is where we differ from Peng and Tang's algorithm, which uses a common threshold function for all items.}
    \item An item appearing at $t=1$ is always rejected.
\end{itemize}

Before we go into the analysis of this family of algorithms, we define the following notations:

$$\bar{p}_i(t)\overset{\underset{\mathrm{def}}{}}{=}\mathbb{P}[v_i>\tau_i(t)] \textrm{\ \ \ \ \ and\ \ \ \ \ } \bar{q}_i(t)\overset{\underset{\mathrm{def}}{}}{=}\mathbb{P}[\max_{j \neq i}v_j>\tau_i(t)].$$
We call a threshold function $\tau_i(t)$ \textit{surjective} if it satisfies $\bar{p}_i(0)=\bar{q}_i(0)=0$ and $\bar{p}_i(1)=\bar{q}_i(1)=1$, i.e. the range of $\tau_i(t)$ contains the support of the distributions. We will only be dealing with such threshold functions. The notations $p_i(t), q_i(t)$ and $\tau(t)$ retain their original meanings from Section \ref{Peng and Tang}. 

 We define an auxiliary function $\bar{g}(t, \Gamma)$ before we define the arrival time distributions.

$$\bar{g}(t, \Gamma) \overset{\underset{\mathrm{def}}{}}{=} 1 - \Gamma \cdot \sum_{i} \int_0^{t}\bar{p}_i(s)\bar{q}_i'(s)ds.$$
We now provide the construction for the arrival time distributions $\{\bar{f}_i(t, \Gamma)\}$. 
$$\bar{f}_i(t, \Gamma) \overset{\underset{\mathrm{def}}{}}{=} \Gamma\cdot\frac{\bar{q}_i'(t)}{\bar{g}(t, \Gamma)}\cdot\exp\left ( -\Gamma \int_{0}^{t}\frac{\bar{q}_i'(s)\bar{p}_i(s)}{\bar{g}(s, \Gamma)}ds\right ).$$
\bigbreak
\noindent \textbf{Condition for the construction to be well defined:} The only condition required for this construction to be well defined is that 
$$\int_0^1\bar{f}_i(t,\Gamma)\cdot dt \leq 1 \textrm{\  for all\  }{i \in [n]}.$$
\bigbreak
We now prove a theorem about the attainment of ASD (refer to Section \ref{ASD}) by the algorithm given that the construction is well defined. Here, the random variable $\ALG$ is the value of the item accepted by the algorithm. 
\begin{theorem}\label{ASD and CR of General Algorithm}
If the arrival time distributions $\{\bar{f}_i(t, \Gamma)\}$ are well defined for some $\Gamma \in (0,1)$ and a set of strictly decreasing continuous \textit{surjective} threshold functions $\{\tau_i(t)\}$ on $t\in[0,1]$, then for all $x > 0$ the following inequality holds
$$\mathbb{P}[\ALG>x] \geq  \Gamma \cdot \mathbb{P}[\max_{i}v_i>x].$$ when the algorithm selects the arrival times of the items from the distributions $\{\bar{f}_i(t, \Gamma)\}$ and uses the threshold functions $\{\tau_i(t)\}$.
\end{theorem}
\begin{proof}
We first define an auxiliary function $l_i(x)$ for each item as follows. $$l_i(x)\overset{\underset{\mathrm{def}}{}}{=} \inf\left(\{t|\tau_i(t) \leq x, t\in [0,1]\} \cup \{1\}\right) \textrm{ for all } x \geq 0.$$
Note that $l_i(x)$ is simply an extended inverse function of $\tau_i(t)$ if $\tau_i(1)=0$. Also note that $l_i(x)$ is continuous.

 We now break down the expression for $\mathbb{P}[\ALG>x]$. Note that we abuse the notation $\ALG$ here to denote both the algorithm itself and the reward collected by the algorithm.

\renewcommand{\arraystretch}{2}
\begin{eqnarray*}
\mathbb{P}[\ALG > x]  & = & \sum_i \mathbb{P}[\textrm{item } i \textrm{ is accepted by ALG\ and } v_i>x ] \\  
 & = & \sum_i \int_{0}^{1}\mathbb{P}[v_i>x, v_i>\tau_i(t), \  \textrm{ALG does not stop before } t | t_i = t]\cdot\bar{f}_i(t, \Gamma)\cdot dt \\ 
 & = & \sum_i \int_{0}^{l_i(x)}\mathbb{P}[v_i>\tau_i(t)]  \cdot  \mathbb{P}[\textrm{ALG does not stop before } t | t_i=t]\cdot\bar{f}_i(t, \Gamma)\cdot dt  \\
 &   & + \sum_i \int_{l_i(x)}^{1}\mathbb{P}[v_i>x]  \cdot \mathbb{P}[\textrm{ALG does not stop before } t | t_i=t]\cdot\bar{f}_i(t, \Gamma)\cdot dt.
\end{eqnarray*}
The final equality follows from the fact that $\tau_i(t)>x$ for $t<l_i(x)$ and $\tau_i(t)<x$ for $t>l_i(x)$. Let us define two functions $A(x)$ and $B(x)$ on $x > 0$.
\begin{equation} \label{A(x)}
A(x) \overset{\underset{\mathrm{def}}{}}{=} \sum_i \int_{0}^{l_i(x)}\mathbb{P}[v_i>\tau_i(t)]  \cdot  \mathbb{P}[\textrm{ALG does not stop before } t | t_i=t]\cdot\bar{f}_i(t, \Gamma)\cdot dt .
\end{equation}
\begin{equation} \label{B(x)}
B(x) \overset{\underset{\mathrm{def}}{}}{=} \sum_i \int_{l_i(x)}^{1}\mathbb{P}[v_i>x]  \cdot \mathbb{P}[\textrm{ALG does not stop before } t | t_i=t]\cdot\bar{f}_i(t, \Gamma)\cdot dt .
\end{equation}
It is clear that $\mathbb{P}[\ALG > x] = A(x)+B(x)$ for $x>0$.
\bigbreak

\begin{lemmma}
For all $t \in [0,1]$ and $i \in [n]$.
$$\mathbb{P}[\textrm{ALG does not stop before } t | t_i=t] =\bar{g}(t, \Gamma) \cdot \exp\left ( \Gamma \int_{0}^{t}\frac{\bar{q}_i'(u)\bar{p}_i(u)}{\bar{g}(u, \Gamma)}du\right ).$$
\end{lemmma}
\begin{proof}
The probability of ALG not stopping before $t$ given $t_i=t$ is simply the probability that no other item arrives at a time $s<t$ with a value greater than its threshold at time $s$. Hence, we have

\begin{equation} \label{ALG does not stop}
\mathbb{P}[\textrm{ALG does not stop before } t | t_i=t] = \prod_{j \neq i} {\left ( 1 - \int_0^{t}\bar{p}_i(s)\bar{f}_i(s, \Gamma)ds \right )}.
\end{equation} 
From the definition of $\bar{f}_i(t, \Gamma)$, we have

\renewcommand{\arraystretch}{2}
\begin{eqnarray*}
\left ( 1 - \int_0^{t}\bar{p}_i(s)\bar{f}_i(s, \Gamma)ds \right) &  =  & 1 - \int_0^t \bar{p}_i(s)\cdot \Gamma \cdot \frac{\bar{q}_i'(s)}{\bar{g}(s, \Gamma)}\cdot\exp\left ( -\Gamma \int_{0}^{s}\frac{\bar{q}_i'(u)\bar{p}_i(u)}{\bar{g}(u, \Gamma)}du\right )ds \\
&  =  & 1+\int_{s=0}^t 1\ d \left(\exp\left ( -\Gamma \int_{0}^{s}\frac{\bar{q}_i'(u)\bar{p}_i(u)}{\bar{g}(u, \Gamma)}du\right )\right)\\
&  =  & 1+\left.\begin{matrix}\exp\left ( -\Gamma \int_{0}^{s}\frac{\bar{q}_i'(u)\bar{p}_i(u)}{\bar{g}(u, \Gamma)}du\right )\end{matrix}\right|_{s=0}^t = \exp\left ( -\Gamma \int_{0}^{t}\frac{\bar{q}_i'(u)\bar{p}_i(u)}{\bar{g}(u, \Gamma)}du\right ).
\end{eqnarray*}
Taking the product of the above expression over all $i \in [n]$, we have

\begin{eqnarray*}\
\prod_{i}{\left ( 1 - \int_0^{t}\bar{p}_i(s)\bar{f}_i(s, \Gamma)ds \right)} &  =  & \prod_i \exp\left ( -\Gamma \int_{0}^{t}\frac{\bar{q}_i'(u)\bar{p}_i(u)}{\bar{g}(u, \Gamma)}du\right ) = \exp\left (  \int_{0}^{t}\frac{-\Gamma\sum_i  \bar{q}_i'(u)\bar{p}_i(u)}{\bar{g}(u, \Gamma)}du\right )\\
&  =  & \exp\left (  \int_{0}^{t}\frac{\bar{g}'(u,\Gamma)}{\bar{g}(u, \Gamma)}du\right ) = \exp\left (  \left.\begin{matrix}\ln (\bar{g}(u, \Gamma))\end{matrix}\right|_{u=0}^t\right ) \\
&  =  & \frac{\bar{g}(t, \Gamma)}{\bar{g}(0, \Gamma)} =   \bar{g}(t, \Gamma). \\
\end{eqnarray*}
Plugging the above results into (\ref{ALG does not stop}), we obtain 
$$\mathbb{P}[\textrm{ALG does not stop before } t | t_i=t] = \frac{\bar{g}(t, \Gamma)}{\left ( 1 - \int_0^{t}\bar{p}_i(s)\bar{f}_i(s, \Gamma)ds \right)}=\bar{g}(t, \Gamma) \cdot \exp\left ( \Gamma \int_{0}^{t}\frac{\bar{q}_i'(u)\bar{p}_i(u)}{\bar{g}(u, \Gamma)}du\right ).$$
\end{proof}
We now proceed to simplify the expressions for $A(x)$ and $B(x)$. From (\ref{A(x)}), we have
\begin{eqnarray*}\
A(x) &  =  & \sum_i \int_{0}^{l_i(x)}\mathbb{P}[v_i>\tau_i(t)]  \cdot  \mathbb{P}[\textrm{ALG does not stop before } t | t_i=t]\cdot\bar{f}_i(t, \Gamma)\cdot dt  \\
&  =  & \sum_i \int_{0}^{l_i(x)}\bar{p}_i(t)  \cdot \bar{g}(t, \Gamma)\cdot \exp\left ( \Gamma \int_{0}^{t}\frac{\bar{q}_i'(u)\bar{p}_i(u)}{\bar{g}(u, \Gamma)}du\right )\cdot\bar{f}_i(t, \Gamma)\cdot dt\\
&  =  & \sum_i \int_{0}^{l_i(x)}\bar{p}_i(t) \bar{g}(t, \Gamma)\exp\left ( \Gamma \int_{0}^{t}\frac{\bar{q}_i'(u)\bar{p}_i(u)}{\bar{g}(u, \Gamma)}du\right )\cdot\Gamma\frac{\bar{q}_i'(t)}{\bar{g}(t, \Gamma)}\exp\left ( -\Gamma \int_{0}^{t}\frac{\bar{q}_i'(s)\bar{p}_i(s)}{\bar{g}(s, \Gamma)}ds\right ) dt \\
&  =  & \Gamma \cdot \sum_i \int_{0}^{l_i(x)} \bar{p}_i(t)   \cdot \bar{q}_i'(t)\cdot dt . \\
\end{eqnarray*}
From (\ref{B(x)}), we have
\begin{eqnarray*}
B(x) &  =  & \sum_i \int_{l_i(x)}^{1}\mathbb{P}[v_i>x]  \cdot \mathbb{P}[\textrm{ALG does not stop before } t | t_i=t]\cdot\bar{f}_i(t, \Gamma)\cdot dt  \\
&  =  & \sum_i \int_{l_i(x)}^{1}\bar{p}_i(l_i(x))  \cdot \bar{g}(t, \Gamma) \cdot \exp\left ( \Gamma \int_{0}^{t}\frac{\bar{q}_i'(u)\bar{p}_i(u)}{\bar{g}(u, \Gamma)}du\right )\cdot\bar{f}_i(t, \Gamma)\cdot dt\\
&  =  & \sum_i \int_{l_i(x)}^{1}\bar{p}_i(l_i(x))  \bar{g}(t, \Gamma) \exp\left ( \Gamma \int_{0}^{t}\frac{\bar{q}_i'(u)\bar{p}_i(u)}{\bar{g}(u, \Gamma)}du\right )\cdot\Gamma\frac{\bar{q}_i'(t)}{\bar{g}(t, \Gamma)}\exp\left ( -\Gamma \int_{0}^{t}\frac{\bar{q}_i'(s)\bar{p}_i(s)}{\bar{g}(s, \Gamma)}ds\right )dt \\
&  =  & \Gamma \cdot \sum_i \bar{p}_i(l_i(x)) \cdot \int_{l_i(x)}^{1}   \bar{q}_i'(t)\cdot dt  = \Gamma \cdot \sum_i \left[\bar{p}_i(l_i(x))  - \bar{p}_i(l_i(x))\cdot \bar{q}_i(l_i(x)) \right] \\
&  =  & \Gamma \cdot \sum_i \int_0^{l_i(x)} \left[ \bar{p}_i'(t)  - (\bar{p}_i(t)\cdot \bar{q}_i(t))' \right]dt \\
& = & \Gamma \cdot \sum_i \int_0^{l_i(x)} \left[ \bar{p}_i'(t)  - \bar{p}_i(t)\cdot \bar{q}_i'(t) - \bar{p}_i'(t)\cdot \bar{q}_i(t) \right]dt .\\
\end{eqnarray*}
The second equality follows from the fact that $\mathbb{P}[v_i>x]=\mathbb{P}[v_i>\tau_i(l_i(x))]=\bar{p}_i(l_i(x))$. Adding the expressions for $A(x)$ and $B(x)$, we get
\begin{eqnarray*}\
A(x)+B(x) &  =  & \Gamma \cdot \sum_i \int_0^{l_i(x)}  \bar{p}_i'(t) (1- \bar{q}_i(t)) dt .\\
\end{eqnarray*}
We obtain the following expression for $\mathbb{P}[\ALG > x]$ for all $x>0$.
\begin{equation} \label{P ALG>x}
\mathbb{P}[\ALG > x] = A(x)+B(x) = \Gamma \cdot \sum_i \int_0^{l_i(x)}  \bar{p}_i'(t) (1- \bar{q}_i(t)) dt .
\end{equation}
We now come up with an expression for $\mathbb{P}[\max_{i}v_i>x]$.
\begin{lemmma}
For all $x>0$,
$$\mathbb{P}[\max_{i}v_i>x]=\sum_i\int_{0}^{l_i(x)}{\bar{p}_i'(t_i)\cdot (1-\bar{q}_i(t_i)) dt_i}.$$
\end{lemmma}
\begin{proof}
Since the distributions are assumed to be continuous, there will almost surely be only one item that has the largest value among all items.

$$\mathbb{P}[\max_{i}v_i>x] = \sum_i\mathbb{P}[ \textrm{Item } i \textrm{ is the maximum value item and } v_i>x].$$ 

Since $\bar{p}_i(t)$ and $l_i(x)$ are differentiable almost everywhere, we can use the expression $\bar{p}_i(l_i(x))-\bar{p}_i(l_i(x+dx)) = -\bar{p}_i'(l_i(x)) \cdot l_i'(x)dx $ to denote the probability of $v_i$ lying in the infinitesimal interval $(x,x+dx)$. Hence, the expression for $\mathbb{P}[\max_{i}v_i>x]$ takes the form 
\begin{eqnarray*}\
\mathbb{P}[\max_{i}v_i>x] &  =  & \sum_i\int_x^\infty{-\bar{p}_i'(l_i(x)) \cdot l_i'(x)\cdot \mathbb{P}[\textrm{Item i is the maximum value item }|v_i=x] dx} \\
 &  =  & \sum_i\int_x^\infty{-\bar{p}_i'(l_i(x)) \cdot l_i'(x)\cdot (1-\bar{q}_i(l_i(x)))\cdot dx}. \\
\end{eqnarray*}
The second equality follows from the definition of $\bar{q}_i(t)$. Now, we perform the change of variables $t_i=l_i(x)$. We have
\begin{eqnarray*}\
\mathbb{P}[\max_{i}v_i>x] &  =  & \sum_i\int_{l_i(x)}^{l_i(\infty)}{-\bar{p}_i'(t_i)\cdot (1-\bar{q}_i(t_i)) dt_i} = \sum_i\int_{0}^{l_i(x)}{\bar{p}_i'(t_i)\cdot (1-\bar{q}_i(t_i))\cdot dt_i}. \\
\end{eqnarray*}
\end{proof}
The above result along with (\ref{P ALG>x}) gives us 
$$\mathbb{P}[\ALG>x] = \Gamma \cdot \mathbb{P}[\max_{i}v_i>x] \textrm{ for all } x>0.$$
 This completes our proof for Theorem \ref{ASD and CR of General Algorithm}.
\end{proof}

\section{The \textit{2-scheme algorithm}}\label{2scheme}
In this section, we describe an algorithm that we call the \textit{2-scheme algorithm}, which attains $\Gamma^*$-ASD over the prophet (Refer to Section \ref{ASD}) for every instance of the order selection prophet inequality problem, where $\Gamma^* = 0.7258$. 
$$\mathbb{P}[\ALG>x] \geq \Gamma^*  \cdot \mathbb{P}[\max_{i}v_i>x] \textrm{ for all } x>0.$$
\textbf{Informal Description: } For a given instance of the problem, the algorithm first checks whether the construction $\{f_i(t,\Gamma^*)\}$ is well defined. Recall that $\{f_i(t,\Gamma^*)\}$ is the set of arrival time distributions used by Peng and Tang's \citep{PT} algorithm. Informally, the \textit{2-scheme algorithm} first attempts to execute Peng and Tang's algorithm with the parameter $\Gamma = \Gamma^*(=0.7258)$. If the constructed arrival time distributions are not well defined, then the algorithm uses an alternate independent arrival time scheme constructed using the framework described in Section \ref{General Framework}. We describe this alternate scheme in detail when we describe the algorithm formally. The claim is that at least one of the two schemes will construct well defined arrival time distributions for $\Gamma = \Gamma^*$.

Before we give a formal description of the algorithm, we define some useful terms.
\begin{definition}[$\Gamma-adverse$ instance, $\Gamma-adverse$ item]\label{adverse}
For some $\Gamma \in (0,1)$, we call an instance of distributions $\{D_i\}$ with $n$ items $\Gamma-adverse$, if there exists at least one item $i \in [n]$ such that

$$\int_0^1 f_i(t, \Gamma)\cdot dt > 1$$
where $\{f_i(t, \Gamma)\}$ are the arrival time distributions constructed by Peng and Tang's algorithm for the instance $\{D_i\}$. Also, we call such an item $\Gamma - adverse$. Note that an item being $\Gamma-adverse$ or not is dependent on the distributions of the remaining items in the instance as well.

\end{definition}
\bigbreak
\noindent \textbf{Formal Description:} We now give a formal description of the algorithm. The algorithm has the option to choose between two Independent Arrival Time schemes, namely \textbf{Scheme I} and \textbf{Scheme II}. If the problem instance is not $\Gamma^* -adverse$, then the algorithm implements \textbf{Scheme I}, and it implements \textbf{Scheme II} otherwise. The descriptions of the 2 schemes follow:
\begin{itemize}
    \item \textbf{Scheme I:} We simply use the arrival time distributions $\{f_i(t, \Gamma^*)\}$, with the threshold function $\tau(t)$ for each item, i.e. we execute Peng and Tang's \citep{PT} algorithm with the parameter $\Gamma=\Gamma^*$
    \item \textbf{Scheme II:} If the instance is $\Gamma^* -adverse$, then there must exist at least one $\Gamma^* -adverse$ item in the instance. Let us say that item $1$ is $\Gamma^*-adverse$ without loss of generality. Now we construct an independent arrival time scheme using the framework provided in the previous section. For $i \neq 1$, we define the threshold functions $\tau_i(t)$ in the following manner:
    $$\tau_i(t) = \tau(t) \textrm{ for all } t \in [0,1].$$
    For $i=1$, we provide a unique construction for the threshold function. The threshold function $\tau_1(t)$ is given by the following equation:
    $$\mathbb{P}[\max_{j\neq 1}v_j>\tau_1(t)] = h(q_1(t)) \textrm{ for } t \in [0,1]$$
    where $h:[0,1]\rightarrow  [0,1]$ is a strictly increasing continuous function satisfying $h(0)=0$ and $h(1)=1$. It is easy to see that $\tau_1(t)$ is \textit{surjective}, continuous and strictly decreasing. The function $h(x)$ is defined in the following manner:
    $$h(x)\overset{\underset{\mathrm{def}}{}}{=}\left\{\begin{matrix}
    \frac{c-\epsilon}{\epsilon}\cdot x & 0 \leq x < \epsilon \\
    c+\epsilon\cdot\frac{x-c}{c-\epsilon} & \epsilon \leq x < c \\
    x & c \leq x \leq 1\\
    \end{matrix}\right.$$
    where $c$ and $\epsilon$ are some constants such that $0 < \epsilon < c < 1$. Specifically, the algorithm uses $c=0.28$ and $\epsilon$ as an arbitrarily small value. Using this construction for the threshold functions $\{\tau_i(t)\}$, the algorithm executes as described in Section \ref{General Framework} with the parameter $\Gamma=\Gamma^*$, i.e. the algorithm uses the arrival time distributions $\{\bar{f}_i(t,\Gamma^*)\}$.
\end{itemize}
We now claim that for every $\Gamma^*-adverse$ instance, the arrival time distributions $\{\bar{f}_i(t,\Gamma^*)\}$ constructed by Scheme II are well defined. Note that in a $\Gamma^*-adverse$ instance, there must be at least one $\Gamma^*-adverse$ item, so it sufficient to discuss the instances where item $1$ is $\Gamma^*-adverse$ without loss of generality.

\begin{restatable}{theorem}{correctness}\label{Correctness of 2-scheme ALG}
For $\Gamma^*=0.7258$, if the following inequality holds for an instance of distributions $\{D_i\}$ with $n$ items, 
$$\int_0^1 f_1(t,\Gamma^*)\cdot dt>1$$
then we have that 
$$\int_0^1 \bar{f}_i(t,\Gamma^*)\cdot dt \leq 1 \textrm{ for all } i \in [n],$$
where $\{\bar{f}_i(t,\Gamma^*)\}$ are the arrival time distributions constructed by Scheme II of the \textit{2-scheme algorithm}.
\end{restatable}
\begin{proof}
The proof for this theorem has been deferred to Appendix \ref{Full proof}.
\end{proof}
Finally, we prove Theorem \ref{Order Selection} by arguing that the \textit{2-scheme algorithm} is $0.7258$-competitive. Recall the statement of Theorem \ref{Order Selection}.
\orderselection*
\begin{proof}
We now argue that the \textit{2-scheme algorithm} is $\Gamma^*$-competitive for all instances of the order selection prophet inequality problem. Consider any instance $\{D_i\}$ of the order selection prophet inequality problem. 

If the instance is not $\Gamma^*-adverse$, then the algorithm executes Scheme I, i.e. it executes Peng and Tang's \citep{PT} algorithm with the parameter $\Gamma=\Gamma^*$. Since the instance is not $\Gamma^*-adverse$, we know that the arrival time distributions $\{f_i(t,\Gamma^*)\}$ constructed by Scheme I are well defined, and from Theorem \ref{ASD and CR of Peng and Tang}, we have that the algorithm attains $\Gamma^*$-ASD over the prophet in this case.

Now let us consider the case when the instance is $\Gamma^*-adverse$. Without loss of generality, let us assume that item $1$ of the instance is $\Gamma^*-adverse$. In this case, the \textit{2-scheme} algorithm executes Scheme II. From Theorem \ref{Correctness of 2-scheme ALG}, we have that the arrival time distributions $\{\bar{f}_i(t,\Gamma^*)\}$ constructed by Scheme II are well defined. Now, from Theorem \ref{ASD and CR of General Algorithm}, we have that the algorithm attains $\Gamma^*$-ASD over the prophet in this case.

It is now clear that the \textit{2-scheme algorithm} attains $\Gamma^*$-ASD over the prophet for all instances of the problem. As mentioned in Section \ref{ASD}, ASD implies competitiveness. From here, it follows that the \textit{2-scheme algorithm} is $\Gamma^*$-competitive.
\end{proof}

\section{The equivalence of competitiveness and ASD} \label{LPproof}

Recall the statement of \textbf{Theorem} \ref{LP}. 
\LP*

Let $\{X_1,\ldots,X_n\}$ be an instance of the order selection prophet inequality problem where $X_1,\ldots,X_n$ are independent random variables taking values from the finite set $\{a_1,\ldots,a_k\}$, where $0=a_1<\cdots<a_k$. Let $X=\max(X_1,\ldots,X_n)$. Consider all algorithms which do the following: pick a permutation $\sigma$ of $\{1,\ldots,n\}$, a sequence of thresholds $\tau_1,\ldots,\tau_n$ and then each $i$ going from $1$ to $n$, if $X_{\sigma(i)}\geq\tau_i$ then accept $X_{\sigma(i)}$ and stop. The behavior of every such algorithm is identical to that of an algorithm which picks the thresholds from the set $\{a_1,\ldots,a_k,\infty\}$. Thus, essentially, there is a finite set $\{A_1,\ldots,A_N\}$ of distinct deterministic algorithms for the order selection prophet problem when the support of the random variables is finite. We overload notation and use $A_i$ to also denote the random variable taking value equal to the value accepted by algorithm $A_i$ when run on $\{X_1,\ldots,X_n\}$.

Consider the following linear program (\hyperref[lpmin]{LP1}).
\begin{samepage}
\begin{framed} \label{lpmin}
$\textbf{min}_{\mu,\boldsymbol{c}}$ $\mu$ subject to
\\\\
$\sum_{j=1}^{k}{c_j\cdot \mathbb{P}[X \geq a_j]}=1$
\\\\
$\mu - \sum_{j=1}^{k}{c_j\cdot \mathbb{P}[A_i \geq a_j]} \geq 0$ for all $i \in [N]$
\\\\
$c_j \geq 0$ for all $j \in [k]$
\end{framed}
\end{samepage}
As a stepping stone for proving Theorem~\ref{LP}, we need the following lemma.

\begin{lemmma} \label{lp_lemma}
Suppose there exists a $\Gamma$-competitive algorithm for all instances of the order selection prophet inequality with finite support size. Then the optimum of the linear program~\hyperref[lpmin]{LP1} is at least $\Gamma$.
\end{lemmma}

\begin{proof}
Let $(c_1,\ldots,c_k)$ be arbitrary non-negative numbers that satisfy the first and the third constraints in linear program~\hyperref[lpmin]{LP1}. We will prove that there necessarily exists an $i^*\in\{1,\ldots,N\}$ such that $\sum_{j=1}^{k}{c_j\cdot \mathbb{P}[A_{i^*} \geq a_j]}\geq\Gamma$. This forces $\mu$ to be at least $\Gamma$ for $(\mu,c_1,\ldots,c_k)$ to be a feasible point.

For $j\in\{1,\ldots,k\}$, define $\theta_j$ to be $c_1+\cdots+c_j$, so that $\theta_1\leq\cdots\leq\theta_k$. Define the function $f:\{a_1,\ldots,a_k\}\longrightarrow\{\theta_1,\ldots,\theta_k\}$ as $f(a_j)=\theta_j$. Observe that $f$ is monotone. Note that some of the $c_j$'s could be zero, which means $\theta_1,\ldots,\theta_k$ need not be all distinct, and hence, $f$ need not be strictly monotone. However, if $c_j>0$, then it implies that $a_j$ is the smallest number which is mapped to $\theta_j$ under $f$.

Let $X'_1,\ldots,X'_n$ be independent random variables defined as $X'_i=f(X_i)$ (recall that $X_i$'s are independent). Let \[X'=\max(X'_1,\ldots,X'_n)=\max(f(X_1),\ldots,f(X_n))=f(\max(X_1,\ldots,X_n))=f(X)\text{,}\]
where the third equality holds because $f$ is monotone. Let $A'$ be a (deterministic) optimal algorithm for the order selection prophet inequality instance $\{X'_1,\ldots,X'_n\}$. Let the random variable denoting the value selected by $A'$ be also denoted by $A'$. Then we have 
$$\mathbb{E}[A']\geq \Gamma\cdot\mathbb{E}[X'].$$
The above inequality arises from our initial assumption that a $\Gamma$-competitive algorithm exists for all finite support instances of the problem. We rewrite it as the following,
\[\sum_{j:c_j>0}\theta_j\cdot\mathbb{P}[A'=\theta_j]\geq\Gamma\sum_{j:c_j>0}\theta_j\cdot\mathbb{P}[X'=\theta_j]\text{.}\]
Since $\theta_j=c_1+c_2+\cdots+c_j$, we have
\begin{equation}\label{eq1}
\sum_{j=1}^kc_j\cdot\mathbb{P}[A'\geq\theta_j]\geq\Gamma\sum_{j=1}^kc_j\cdot\mathbb{P}[X'\geq\theta_j]\text{.}
\end{equation}

Define the algorithm $A$ for the order selection prophet inequality problem on $\{X_1,\ldots,X_n\}$ to be the following. Order $X_1,\ldots,X_n$ exactly as $A'$ orders $X'_1,\ldots,X'_n$. On receiving a sample $x_j$ of $X_j$, pass $f(x_j)$ to $A'$, and accept $x_j$ if and only if $A'$ accepts $f(x_j)$. Since $A'$ is a deterministic algorithm, so is $A$. Therefore, the behavior of $A$ is identical to that of $A_{i^*}$ for some $i^*\in\{1,\ldots,N\}$. Thus, for every $j$ such that $c_j>0$, we have $A_{i^*}\geq a_j$ if and only if $A'\geq\theta_j$, where the ``if'' part follows from the fact that $a_j$ is the smallest number mapped to $\theta_j$ under $f$. Thus,
\begin{equation}\label{eq2}
\sum_{j:c_j>0}c_j\cdot\mathbb{P}[A_{i^*}\geq a_j]=\sum_{j:c_j>0}c_j\cdot\mathbb{P}[A'\geq\theta_j]\text{.}
\end{equation}

Now, we have
\begin{eqnarray*}
\sum_{j=1}^kc_j\cdot\mathbb{P}[A_{i^*}\geq a_j] & = & \sum_{j:c_j>0}c_j\cdot\mathbb{P}[A_{i^*}\geq a_j]=\sum_{j:c_j>0}c_j\cdot\mathbb{P}[A'\geq\theta_j]=\sum_{j=1}^kc_j\cdot\mathbb{P}[A'\geq\theta_j]\\
 & \geq & \Gamma\sum_{j=1}^kc_j\cdot\mathbb{P}[X'\geq\theta_j]\geq\Gamma\sum_{j=1}^kc_j\cdot\mathbb{P}[X\geq a_j]=\Gamma\text{,}
\end{eqnarray*}
as required, where the second equality is same as equation~\ref{eq2}, the first inequality is same as inequality~\ref{eq1}, the second inequality follows from the fact that $X'=f(X)$ and $f$ is monotone, and the last equality follows because $(c_1,\ldots,c_k)$ satisfies the first constraint of the linear program~\hyperref[lpmin]{LP1}.
\end{proof}

Let us construct the dual of the linear program~\hyperref[lpmin]{LP1}. We shall refer to this dual linear program as \hyperref[lpmax]{LP2}.
\begin{samepage}

\begin{framed} \label{lpmax}
$\textbf{max}_{\alpha,\boldsymbol{\lambda}}$ $\alpha$ subject to
\\\\
$\sum_{i=1}^{N}{\lambda_i\cdot \mathbb{P}[A_i \geq a_j]} \geq \alpha \cdot \mathbb{P}[X \geq a_j]$ for all $j \in [k]$
\\\\$\lambda_i \geq 0$ for all $i \in [N]$
\\\\
$\sum_{i=1}^N{\lambda_i}=1$
\end{framed}
\end{samepage}
We now have the tools required for proving Theorem \ref{LP}.
Let us suppose that there exists a $\Gamma$-competitive algorithm for all instances of the order selection prophet inequality with finite support size. By Lemma~\ref{lp_lemma} and strong LP duality, the optimum of the linear program~\hyperref[lpmax]{LP2} is at least $\Gamma$. Let $(\alpha^*,\lambda^*_1,\ldots,\lambda^*_N)$ be an optimal feasible point, so that $\alpha^*\geq\Gamma$. Consider the probability distribution on the set $\{A_1,\ldots,A_N\}$ which assigns probability mass $\lambda^*_i$ to the algorithm $A_i$. Let $A$ be a random algorithm sampled from this distribution, and let $A$ also denote the random variable which takes value equal to the value accepted by the algorithm $A$. We claim that algorithm $A$ is $\Gamma$-ASD. Observe that since $X_1,\ldots,X_n$, $X=\max(X_1,\ldots,X_n)$, and $A$ take values from the set $\{a_1,\ldots,a_k\}$, it is sufficient to check the approximate stochastic dominance condition only for $a_1,\ldots,a_k$. Indeed, we have,
\[\mathbb{P}[A\geq a_j]=\sum_{i=1}^{N}{\lambda^*_i\cdot \mathbb{P}[A_i \geq a_j]} \geq \alpha^* \cdot \mathbb{P}[X \geq a_j]\geq \Gamma \cdot \mathbb{P}[X \geq a_j]\text{,}\]
as required.

\begin{remark}
It is worth noting that just by changing the definition of a "deterministic algorithm" in the above proof, we can show the same result for a wide range of arrival order settings, including the \textit{order-aware} prophet secretary, the \textit{order-unaware} prophet secretary problem and the constrained order prophet inequality \citep*{constrained}. More specifically, for these problems, if there exists a $\Gamma$-competitive algorithm for all finite support instances, then there also exists a $\Gamma$-ASD algorithm for all finite support instances. For example, in the proof of this claim for the \textit{order-aware} prophet secretary, the finite set of deterministic algorithms will consist of all mappings from the set of arrival orders to the set of sequences of thresholds.
\end{remark}

\section{The Random Order Setting (Prophet Secretary)} \label{Random Order Section}
In this section, we prove Theorem \ref{Random Order}.
\randomorder*
\begin{proof}
We show this by constructing an instance of the prophet secretary problem and arguing that even the optimal algorithm can not achieve a competitive ratio greater than 0.7254 on that instance.
\bigbreak

\noindent \textbf{Construction:} The instance is composed of $N+1$ items, where 1 item is deterministic and has the value $a>0$. The remaining $N$ items are IIDs supported on $\{0,b_1,b_2,\ldots,b_k,\frac{1}{N\cdot \epsilon}\}$ for some $k>0$ and $\epsilon>0$ with $a<b_1<b_2<\cdots<b_k<\frac{1}{N\cdot\epsilon}$. These IID variables share the following distribution for some $p_i>0$ for $i \in [k]$:
$$X=\left\{\begin{matrix}
0 & w.p. &  1-\sum_{i=1}^k \frac{p_i}{N}-\epsilon\\
b_i & w.p. &  \frac{p_i}{N}, \textrm{ for } i \in [k]\\
\frac{1}{N\cdot \epsilon} & w.p. &  \epsilon.\\
\end{matrix}\right. $$
We now look at the distribution of $\max_i{v_i}$, i.e. the maximum of the values drawn from all $N+1$ distributions. 
$$\mathbb{P}\left[\max_i{v_i}=\frac{1}{N\cdot \epsilon}\right]=1-(1-\epsilon)^N,$$
and for all $i \in [k]$
$$\mathbb{P}\left[\max_i{v_i}=b_i\right]=\left(1-\epsilon-\sum_{j=i+1}^k \frac{p_j}{N}\right)^N-\left(1-\epsilon-\sum_{j=i}^k\frac{p_j}{N}\right)^N.$$
We also have
$$\mathbb{P}\left[\max_i{v_i}=a\right]=\left(1-\epsilon-\sum_{j=1}^k\frac{p_j}{N}\right)^N.$$
From here we have
$$\mathbb{E}[\max_i{v_i}] = a\cdot \left(1-\epsilon-\sum_{j=1}^k\frac{p_j}{N}\right)^N + \sum_{i=1}^{k} b_i \cdot \left[\left(1-\epsilon-\sum_{j=i+1}^k \frac{p_j}{N}\right)^N-\left(1-\epsilon-\sum_{j=i}^k\frac{p_j}{N}\right)^N\right]+\frac{[1-(1-\epsilon)^N]}{N\cdot \epsilon}.$$
We shall denote this expected value by $\textrm{MAX}(N,k,a,\mathbf{b},\mathbf{p},\epsilon)$. The following expression will be of use later,
\begin{equation}\label{MAX}
\lim_{\epsilon\rightarrow 0^+}\textrm{MAX}(N,k,a,\mathbf{b},\mathbf{p},\epsilon) = a\cdot \left(1-\sum_{j=1}^k\frac{p_j}{N}\right)^N + \sum_{i=1}^{k} b_i \cdot \left[\left(1-\sum_{j=i+1}^k \frac{p_j}{N}\right)^N-\left(1-\sum_{j=i}^k\frac{p_j}{N}\right)^N\right]+1.
\end{equation}
\bigbreak
\noindent \textbf{The Optimal Algorithm:} Suppose that the items arrive in a fixed order $D_1, D_2, D_3, \ldots, D_n$. It is well understood how the optimal algorithm operates in such a scenario. Suppose that for some $1\leq i<n$, the optimal algorithm is currently at the $i$-th item, whose realised value is $v_i\sim D_i$. The optimal algorithm accepts the item if and only if $v_{i}$ exceeds the expected reward of the optimal algorithm if it were to reject it. This leads to the relation
$$\tau_i=\tau_{i+1}+\mathbb{E}[(v_i-\tau_{i+1})^+] \ \ \ \textrm{for \ \ \ } 1\leq i < n$$
where $\tau_i$ denotes the expected reward of the optimal algorithm given that it starts at the $i$-th item. It is clear that $\tau_n=\mathbb{E}[v_n]$, where $v_n \sim D_n$. The optimal algorithm's expected reward for this fixed ordering is given by $\tau_1$. 

In our instance, all items except one are identically distributed, and hence there are only $N+1$ distinct orderings in which the items can arrive. For each ordering, we can use the above recursive relation and numerically compute the optimal algorithm's expected reward. Since each ordering is equally likely to occur, we simply take the mean of the optimal algorithm's expected reward for each ordering to obtain the optimal algorithm's expected reward in the random order setting. We denote this expected reward by $\OPT(N,k,a,\mathbf{b},\mathbf{p},\epsilon)$, since this value is dependent on the values of the parameters of the instance. 
For any given $N,k,a,\mathbf{b}$ and $\mathbf{p}$, it is easy to compute the value of $\lim_{\epsilon\rightarrow0^+}\OPT(N,k,a,\mathbf{b},\mathbf{p},\epsilon)$ numerically, since the recursive relation described above can be simplified using the following expression,
$$\lim_{\epsilon\rightarrow 0^+}\mathbb{E}[(v_i-\tau_{i+1})^+] = \sum_{j:b_j>\tau_{i+1}} (b_j-\tau_{i+1})\cdot\frac{ p_j}{N}+\frac{1}{N},$$
when the $i$-th item is one of the IID variables. The case when the $i$-th item is the deterministic variable is handled trivially.

Fixing the following values for the parameters of our construction, $$N=10^5,\ \  k=12,\ \  a=0.82,$$ 
\begin{center}

\begin{tabular}{||c | c c c c c c c c c c c c||}  
  \hline
  i & 1 & 2 & 3 & 4 & 5 & 6 & 7 & 8 & 9 & 10 & 11 & 12 \\
  \hline\hline
  $b_i$ & 1.2 & 1.25 & 1.3 & 1.35 & 1.4 & 1.45 & 1.5 & 1.55 & 1.6 & 1.65 & 1.7 & 1.8 \\ 
  \hline
  $p_i$ & 0.02 & 0.03 & 0.04 & 0.05 & 0.04 & 0.03 & 0.03 & 0.02 & 0.02 & 0.02 & 0.02 & 0.005 \\ 
  \hline
\end{tabular}    
\end{center}
we observe that 
$$\frac{\lim_{\epsilon\rightarrow0^+}\OPT(N,k,a,\mathbf{b},\mathbf{p},\epsilon)}{\lim_{\epsilon\rightarrow0^+}\textrm{MAX}(N,k,a,\mathbf{b},\mathbf{p},\epsilon)} \approx 0.725398 < 0.7254 ,$$
where $\lim_{\epsilon\rightarrow 0^+}\OPT(N,k,a,\mathbf{b},\mathbf{p},\epsilon)$ is evaluated numerically using the technique described above, and $\lim_{\epsilon\rightarrow 0^+}\textrm{MAX}(N,k,a,\mathbf{b},\mathbf{p},\epsilon)$ is evaluated using Equation \ref{MAX}. Observe that
$$\lim_{\epsilon\rightarrow 0^+}\left(\frac{\OPT(N,k,a,\mathbf{b},\mathbf{p},\epsilon)}{\textrm{MAX}(N,k,a,\mathbf{b},\mathbf{p},\epsilon)}\right)=\frac{\lim_{\epsilon\rightarrow0^+}\OPT(N,k,a,\mathbf{b},\mathbf{p},\epsilon)}{\lim_{\epsilon\rightarrow0^+}\textrm{MAX}(N,k,a,\mathbf{b},\mathbf{p},\epsilon)}<0.7254.$$
This holds because both the limits are non-zero and finite. From the definition of limits, we now have that there exists an $\epsilon>0$ for which $$\frac{\OPT(N,k,a,\mathbf{b},\mathbf{p},\epsilon)}{\textrm{MAX}(N,k,a,\mathbf{b},\mathbf{p},\epsilon)}<0.7254.$$
This completes our proof for Theorem \ref{Random Order}.

\end{proof}

\appendix
\section{Proof of Theorem \ref{Correctness of 2-scheme ALG}: Correctness of the \textit{2-scheme algorithm}} \label{Full proof}
This section of the appendix is dedicated to proving Theorem \ref{Correctness of 2-scheme ALG}. Recall the statement of the theorem.
\correctness*

\subsection{Setting up some tools and notation}
 In this section, the functions $\bar{p}_i(t)$, $\bar{q}_i(t)$ and $\bar{g}(t, \Gamma)$ are defined with respect to the threshold functions constructed in Scheme II of the \textit{2-scheme algorithm}. The functions $p_i(t)$, $q_i(t)$ and $g(t, \Gamma)$ retain their original meanings, i.e. they are defined with respect to the common threshold function $\tau(t)$ used by Scheme I.

 We define $q_i^{-1}(t)$ and $\bar{q}_i^{-1}(t)$ as the inverses of  $q_i(t)$ and $\bar{q}_i(t)$ respectively. Since we have assumed $q_i(t)$ to be strictly increasing with $q_i(0)=0$ and $q_i(1)=1$, the inverse is well defined and has those same properties. We can say the same about the functions $\bar{q}_1(t)=h(q_1(t))$, and $\bar{q}_i(t)=q_i(t)$ for $i \neq 1$. We also define the following functions:
$$\tilde{p}_i(x)\overset{\underset{\mathrm{def}}{}}{=}p_i(q_i^{-1}(x)) \textrm{\ \ \ and \ \ \ }{\tilde{\bar{p}}_i(x)\overset{\underset{\mathrm{def}}{}}{=}\bar{p}_i(\bar{q}_i^{-1}(x))}.$$
From the definitions, it is easy to see that $\tilde{p}_i(x) = \tilde{\bar{p}}_i(x)$ for all $x\in [0,1]$. We also define:
$$\tilde{g}_i(x,\Gamma)\overset{\underset{\mathrm{def}}{}}{=}g(q_i^{-1}(x),\Gamma) \textrm{\ \ \ and \ \ \ }{\tilde{\bar{g}}_i(x,\Gamma)\overset{\underset{\mathrm{def}}{}}{=}\bar{g}(\bar{q}_i^{-1}(x),\Gamma)}.$$
We have
$$\int_0^1 f_i(t, \Gamma)\cdot dt = \int_0^1 \frac{\Gamma \cdot q_i'(t)}{g(t, \Gamma)\exp\left(\Gamma \cdot \int_0^t\frac{p_i(s)q_i'(s)}{g(s, \Gamma)}ds\right)}dt.$$
We perform the change of variables x=$q_i(t)$, y=$q_i(s)$.
\begin{equation}\label{f_i(t) with changed variables}
\int_0^1 f_i(t, \Gamma)\cdot dt = \int_0^1 \frac{\Gamma}{\tilde{g}_i(x, \Gamma)\exp\left(\Gamma \cdot \int_0^x\frac{\tilde{p}_i(y)}{\tilde{g}_i(y, \Gamma)}dy\right)}dx.
\end{equation}
Similarly, we obtain
\begin{equation}\label{bar f_i(t) with changed variables}
\int_0^1 \bar{f}_i(t, \Gamma)\cdot dt = \int_0^1 \frac{\Gamma}{\tilde{\bar{g}}_i(x, \Gamma)\exp\left(\Gamma \cdot \int_0^x\frac{\tilde{p}_i(y)}{\tilde{\bar{g}}_i(y, \Gamma)}dy\right)}dx.
\end{equation}
Note that we have used $\tilde{p}_i(x) = \tilde{\bar{p}}_i(x)$ to obtain the above expression. We now state a lemma from \citet{PT}.
\begin{lemmma}\label{g tilde bound}
For all $i \in [n]$ and $x \in [0,1)$,
$$\tilde{g}_i(x,\Gamma) \geq \Gamma\cdot [-(1-x)\cdot \ln(1-x)\cdot (1-\tilde{p}_i(x))-x]+1.$$.
\end{lemmma}
 We now prove the following lemma about $g(t,\Gamma)$.
\begin{lemmma}\label{g(t) definition}
For all $t \in [0,1]$,
$$g(t,\Gamma) = 1 - \Gamma \cdot \sum_i \int_0^t p_i(s)q_i'(s)ds.$$.
\end{lemmma}
\begin{proof}
We have that 
$$g'(t,\Gamma) = \left(\Gamma \cdot \left(\sum_i (1-q_i(t))\cdot p_i(t)-t \right)+1 \right)',$$
where $g'(t,\Gamma)=\frac{d}{dt}g(t,\Gamma)$. The above expression simplifies to
\begin{equation}\label{g'(t)}
g'(t,\Gamma) = -\Gamma\cdot \sum_i q_i'(t)p_i(t) + \Gamma\cdot \sum_i (1-q_i(t))\cdot p_i'(t) - \Gamma.
\end{equation}
From the definitions of $p_i(t)$ and $q_i(t)$, we have
$$\prod_i (1-p_i(t))=1-t \textrm{\ \ \ and \ \ \ } \prod_{j\neq i} (1-p_j(t))=1-q_i(t)$$
for all $i\in[n]$, $t\in[0,1]$. Taking derivative of $-\prod_i (1-p_i(t))$, we obtain
$\sum_i (1-q_i(t))\cdot p_i'(t) = 1$.
Combining this with (\ref{g'(t)}), we obtain
$$g'(t, \Gamma) = -\Gamma\cdot \sum_i q_i'(t)p_i(t).$$
Applying the boundary condition $g(0,\Gamma)=1$, we obtain
$g(t,\Gamma) = 1 - \Gamma \cdot \sum_i \int_0^t p_i(s)q_i'(s)ds$.
\end{proof}
\begin{corollary}\label{bound on q_i^-1(q_i())}
For all $i \in [n]$ and $x \in [0,1)$,
$$\sum_j \int_0^{q_j(q_i^{-1}(x))} \tilde{p}_j(s)\cdot ds \leq - [-(1-x)\cdot \ln(1-x)\cdot (1-\tilde{p}_i(x))-x].$$
\end{corollary}
\begin{proof}
The following statement follows directly from Lemmas \ref{g tilde bound} and \ref{g(t) definition}.
$$1-\Gamma \cdot \sum_j \int_0^{q_i^{-1}(x)} p_j(t)\cdot q_j'(t)\cdot dt \geq \Gamma\cdot [-(1-x)\cdot \ln(1-x)\cdot (1-\tilde{p}_i(x))-x]+1.$$
Performing the change of variable $s=q_j(t)$ in each term of the summation on the LHS, we get
$$1-\Gamma \cdot \sum_j \int_0^{q_j(q_i^{-1}(x))} \tilde{p}_j(s)\cdot ds \geq \Gamma\cdot [-(1-x)\cdot \ln(1-x)\cdot (1-\tilde{p}_i(x))-x]+1.$$
The statement of the corollary follows directly from here.
\end{proof}

Recall the construction of threshold functions used by Scheme II. For $i\neq 1$, we have
$$\tau_i(t) = \tau(t) \textrm{ for all } t \in [0,1] \textrm{ and } i \in [n]-\{1\}.$$
This means that the following holds.
\begin{equation}\label{q_i(t)}
    \bar{q}_i(t)=q_i(t) \textrm{\ for \ } i \in [n]-\{1\}.
\end{equation}
For $i=1$, we have a different threshold function. $\tau_1(t)$ is given by the following equation:
\begin{equation}\label{q_1(t)}
\bar{q}_1(t)=\mathbb{P}[\max_{j\neq 1}v_j>\tau_1(t)] = h(q_1(t)) \textrm{ for } t \in [0,1].
\end{equation}
Also, recall the definition of $h(x)$. 
$$h(x)\overset{\underset{\mathrm{def}}{}}{=}\left\{\begin{matrix}
    \frac{c-\epsilon}{\epsilon}\cdot x & 0 \leq x < \epsilon \\
    c+\epsilon\cdot\frac{x-c}{c-\epsilon} & \epsilon \leq x < c \\
    x & c \leq x \leq 1.\\
    \end{matrix}\right.$$
We now prove two useful lemmas about functions $\tilde{\bar{g}}_i(x,\Gamma)$.
\begin{lemmma}\label{Bound on g1 tilde}
For all $x \in [0,c-\epsilon)$,
$$\tilde{\bar{g}}_1(x,\Gamma) \geq 1-\Gamma\cdot\int_0^x\tilde{p}_1(s)\cdot ds - \Gamma\cdot \epsilon$$
and for all $x \in [c,1)$,
$$\tilde{\bar{g}}_1(x,\Gamma) \geq \Gamma\cdot [-(1-x)\cdot \ln(1-x)\cdot (1-\tilde{p}_1(x))-x]+1 .$$
\end{lemmma}
\begin{proof}
We perform the change of variables $u=\bar{q}_i(s)$ in each term of the summation in the definition of $\bar{g}(t,\Gamma)$.
$$\bar{g}(t,\Gamma)= 1 - \Gamma \cdot \sum_{i} \int_0^{t}\bar{p}_i(s)\bar{q}_i'(s)\cdot ds = 1 - \Gamma \cdot \sum_{i } \int_0^{\bar{q}_i(t)}\tilde{p}_i(u)\cdot du.$$
From the definition of $\tilde{\bar{g}}_i(x,\Gamma)$, we have
$$\tilde{\bar{g}}_1(x,\Gamma)= \bar{g}(\bar{q}_1^{-1}(x),\Gamma)=1 - \Gamma \cdot \sum_{i } \int_0^{\bar{q}_i(\bar{q}_1^{-1}(x))}\tilde{p}_i(u)\cdot du.$$
Recall that $\bar{q}_1(t)=h(q_1(t))$, and hence
$\bar{q}_1^{-1}(x)=q_1^{-1}(h^{-1}(x))$.
We now resolve the expression for $\tilde{\bar{g}}_1(x,\Gamma)$.
\begin{eqnarray*}\
\tilde{\bar{g}}_1(x,\Gamma)  & = & 1 - \Gamma \cdot \sum_{i } \int_0^{\bar{q}_i(q_1^{-1}(h^{-1}(x)))}\tilde{p}_i(u)\cdot du\\  
& = & 1 - \Gamma \cdot \sum_{i \neq 1} \int_0^{\bar{q}_i(q_1^{-1}(h^{-1}(x)))}\tilde{p}_i(u)\cdot du - \Gamma \cdot \int_0^{\bar{q}_1(q_1^{-1}(h^{-1}(x)))}\tilde{p}_1(u)\cdot du\\ 
& = & 1 - \Gamma \cdot \sum_{i \neq 1} \int_0^{q_i(q_1^{-1}(h^{-1}(x)))}\tilde{p}_i(u)\cdot du - \Gamma \cdot \int_0^{\bar{q}_1(q_1^{-1}(h^{-1}(x)))}\tilde{p}_1(u)\cdot du\\
& = & 1 - \Gamma \cdot \sum_{i} \int_0^{q_i(q_1^{-1}(h^{-1}(x)))}\tilde{p}_i(u)\cdot du - \Gamma \cdot \int_{q_1(q_1^{-1}(h^{-1}(x)))}^{\bar{q}_1(q_1^{-1}(h^{-1}(x)))}\tilde{p}_1(u)\cdot du\\
& = & 1 - \Gamma \cdot \sum_{i} \int_0^{q_i(q_1^{-1}(h^{-1}(x)))}\tilde{p}_i(u)\cdot du - \Gamma \cdot \int_{h^{-1}(x)}^{x}\tilde{p}_1(u)\cdot du.\\
\end{eqnarray*}
Applying Corollary \ref{bound on q_i^-1(q_i())}, we have:
$$\tilde{\bar{g}}_1(x,\Gamma) \geq \Gamma\cdot [-(1-h^{-1}(x))\cdot \ln(1-h^{-1}(x))\cdot (1-\tilde{p}_i(h^{-1}(x)))-h^{-1}(x)]+1 - \Gamma \cdot \int_{h^{-1}(x)}^{x}\tilde{p}_1(u)\cdot du.$$
From the definition of $h(x)$, we have that $h^{-1}(x)< \epsilon$ for $x<c-\epsilon$, and $h^{-1}(x)=x$ for $x\geq c$. Hence,
for all $x \in [0,c-\epsilon)$,
$$\tilde{\bar{g}}_1(x,\Gamma) \geq 1-\Gamma\cdot\int_0^x\tilde{p}_1(s)\cdot ds - \Gamma\cdot\epsilon$$
and for all $x \in [c,1)$,
$$\tilde{\bar{g}}_1(x,\Gamma) \geq \Gamma\cdot [-(1-x)\cdot \ln(1-x)\cdot (1-\tilde{p}_1(x))-x]+1 .$$
\end{proof}

\begin{lemmma}\label{Bound on gi tilde}
For $i\neq 1$ and $x \in [0,1)$,
$$\tilde{\bar{g}}_i(x,\Gamma) \geq \Gamma\cdot [-(1-x)\cdot \ln(1-x)\cdot (1-\tilde{p}_i(x))-x]+1-\Gamma\cdot\int_0^c\tilde{p}_1(s)\cdot ds.$$
\end{lemmma}
\begin{proof}
From the definition of $\tilde{\bar{g}}_i(x,\Gamma)$, we have
$$\tilde{\bar{g}}_i(x,\Gamma)= \bar{g}(\bar{q}_i^{-1}(x),\Gamma)=1 - \Gamma \cdot \sum_{j } \int_0^{\bar{q}_j(\bar{q}_i^{-1}(x))}\tilde{p}_j(u)\cdot du.$$
Applying (\ref{q_i(t)}) and (\ref{q_1(t)}),
\begin{eqnarray*}\
\tilde{\bar{g}}_i(x,\Gamma) &  =  & 1 - \Gamma \cdot \sum_{j\neq 1 } \int_0^{q_j(q_i^{-1}(x))}\tilde{p}_j(u)\cdot du - \Gamma \cdot \int_0^{\bar{q}_1(q_i^{-1}(x))}\tilde{p}_1(u)\cdot du\\  
&  =  &  1 - \Gamma \cdot \sum_{j } \int_0^{q_j(q_i^{-1}(x))}\tilde{p}_j(u)\cdot du - \Gamma \cdot \int_{q_1(q_i^{-1}(x))}^{\bar{q}_1(q_i^{-1}(x))}\tilde{p}_1(u)\cdot du\\  
&  =  &  1 - \Gamma \cdot \sum_{j } \int_0^{q_j(q_i^{-1}(x))}\tilde{p}_j(u)\cdot du - \Gamma \cdot \int_{q_1(q_i^{-1}(x))}^{h(q_1(q_i^{-1}(x)))}\tilde{p}_1(u)\cdot du.\\ 
\end{eqnarray*}
From the definition of $h(x)$, notice that the maximum value that $\int_{y}^{h(y)}\tilde{p}_1(u)\cdot du$ can take for some $y\in[0,1]$ is no more than $\int_{0}^{c}\tilde{p}_1(u)\cdot du$. Using this fact and applying Corollary \ref{bound on q_i^-1(q_i())}, we have for all $x \in [0,1)$:
$$\tilde{\bar{g}}_i(x,\Gamma) \geq \Gamma\cdot [-(1-x)\cdot \ln(1-x)\cdot (1-\tilde{p}_i(x))-x]+1-\Gamma\cdot\int_0^c\tilde{p}_1(s)\cdot ds.$$
\end{proof}
We now define a new notation $\hat{g}_i(x, \Gamma)$.
$$\hat{g}_i(x, \Gamma) \overset{\underset{\mathrm{def}}{}}{=} \Gamma\cdot [-(1-x)\cdot \ln(1-x)\cdot (1-\tilde{p}_i(x))-x]+1.$$
From Lemma \ref{g tilde bound}, we have $\tilde{g}_i(x,\Gamma) \geq \hat{g}_i(x, \Gamma)$ for all $i\in [n]$ and $x \in [0,1)$. We also define two notations $G_i(z,\Gamma)$ and $\hat{G}_i(z,\Gamma)$.
$$G_i(z,\Gamma)\overset{\underset{\mathrm{def}}{}}{=}\int_z^1 \frac{\Gamma}{\tilde{g}_i(x, \Gamma)\exp\left(\Gamma \cdot \int_z^x\frac{\tilde{p}_i(y)}{\tilde{g}_i(y, \Gamma)}dy\right)}dx \ \ \ \textrm{and} \ \ \ \hat{G}_i(z,\Gamma)\overset{\underset{\mathrm{def}}{}}{=}\int_z^1 \frac{\Gamma}{\hat{g}_i(x, \Gamma)\exp\left(\Gamma \cdot \int_z^x\frac{\tilde{p}_i(y)}{\hat{g_i}(y, \Gamma)}dy\right)}dx.$$
It is clear from (\ref{f_i(t) with changed variables}) that 
$$\int_0^1 f_i(t,\Gamma) = G_i(0,\Gamma).$$
Hence, an item is $\Gamma-adverse$ (recall \textbf{Definition} \ref{adverse}) if and only if $G_i(0,\Gamma)>1$.

\begin{restatable}{lemmma}{rhoRho}\label{rho1>rho2}
Suppose that for functions $\tilde{p}, \rho_1,\rho_2:[0,1]\rightarrow [0,1]$, such that $\tilde{p}$ is strictly increasing on $[0,1]$, we have
$$\rho_1(x)\geq \rho_2(x)>0$$
for all $x \in [0,1]$, then
$$\int_0^1 \frac{\Gamma}{{\rho}_1(x)\exp\left(\Gamma \cdot \int_0^x\frac{\tilde{p}(y)}{{\rho}_1(y)}dy\right)}dx \leq 1 \textrm{\ \ \ if \ \ \ } \int_0^1 \frac{\Gamma}{{\rho}_2(x)\exp\left(\Gamma \cdot \int_0^x\frac{\tilde{p}(y)}{{\rho}_2(y)}dy\right)}dx \leq 1.$$
\end{restatable}
\begin{proof}
We defer the proof for this lemma to Appendix \ref{proof of rho1>rho2}.
\end{proof}

It follows from Lemma \ref{rho1>rho2} that $\hat{G}_i(0,\Gamma) \leq 1$ implies $G_i(0,\Gamma) \leq 1$. We now define a weaker version of $\Gamma-adverse$-ness.

\begin{definition}[$\Gamma-weakly-adverse$ instance, $\Gamma-weakly-adverse$ item]
For some $\Gamma \in (0,1)$, we call an instance of distributions $\{D_i\}$ with $n$ items $\Gamma-weakly-adverse$, if there exists at least one item $i \in [n]$ such that

$$\hat{G}_i(0,\Gamma) > 1.$$
Also, we call such an item $\Gamma -weakly-adverse$. Note that an item being $\Gamma-weakly-adverse$ or not is dependent on the distributions of the remaining items in the instance as well.
\end{definition}
 It is easy to see that every $\Gamma-adverse$ instance/item is $\Gamma-weakly-adverse$ as well.
\begin{restatable}{lemmma}{ordering}\label{Ordering on gamma}
For $\Gamma_1,\Gamma_2 \in (0,1)$ (with $\Gamma_1<\Gamma_2$), if an item $i$ in an instance of distributions $\{D_i\}$ is $\Gamma_1-weakly-adverse$, then item $i$ is also $\Gamma_2-weakly-adverse$.
\end{restatable}
\begin{proof}
We defer the proof for this lemma to Appendix \ref{proof of rho1>rho2}.
\end{proof}

\subsection{Auxiliary Functions}

We now define some auxiliary functions that will be useful in the proof. Some of these definitions have been borrowed from \citet{PT}. For $z\in[0,1)$ and $\Gamma \in (0,1)$, let

\begin{itemize}
    \item $H(z, \Gamma) \overset{\underset{\mathrm{def}}{}}{=} \frac{\Gamma(-(1-z)\cdot\ln(1-z))}{\Gamma[-(1-z)\cdot\ln(1-z)-z]+1}$
    \item $K(z, \Gamma) \overset{\underset{\mathrm{def}}{}}{=} \frac{\Gamma \cdot(1-z)}{1-\Gamma \cdot z}$
    \item $M(z, \Gamma) \overset{\underset{\mathrm{def}}{}}{=} 1 - \int_0^z\frac{\Gamma}{\Gamma \cdot [-(1-x)\cdot\ln(1-x)-x]+1}dx$
\end{itemize}
We will use the shorthand notations $H_\Gamma(z)$, $K_\Gamma(z)$ and $M_\Gamma(z)$ to denote $H(z, \Gamma)$, $K(z, \Gamma)$ and $M(z, \Gamma)$. We also define the following values that are dependent on $\Gamma$.
\begin{itemize}
    \item $\beta_\Gamma\overset{\underset{\mathrm{def}}{}}{=}\inf\{z|M(z,\Gamma)\leq H(z,\Gamma), z\in[0,1)\}$
    \item $\gamma_\Gamma$ is defined as the unique root of $H(z,\Gamma)=K(z,\Gamma)$ on $z\in[0,1)$. We defer the proof for the uniqueness of this root to Appendix \ref{Uniqueness of gamma}.
\end{itemize} 
We now define a few more useful functions in terms of the auxiliary functions defined above.
\begin{itemize}
    \item $\mu_\Gamma(x)\overset{\underset{\mathrm{def}}{}}{=}[H_\Gamma(\gamma_\Gamma)-M_\Gamma(\gamma_\Gamma)]\cdot\frac{\Gamma[-(1-x)\cdot\ln(1-x)-x]+1}{\Gamma\cdot(M_\Gamma(x)-H_\Gamma(x))}  $
for $x\in[0,\beta_\Gamma)$
    \item $Y_\Gamma(z,p)\overset{\underset{\mathrm{def}}{}}{=}\left[\frac{\Gamma \cdot (1-K_\Gamma(z))}{1-\Gamma\cdot z}-\frac{\Gamma \cdot (1-p\cdot K_\Gamma(z))}{\Gamma\cdot [-(1-z)\cdot\ln(1-z)\cdot(1-p)-z]+1}\right]\cdot(1-\Gamma z)$
for $z\in[\gamma_\Gamma,1),\ p \in [0,1]$ 
    \item $W_\Gamma(z,p)\overset{\underset{\mathrm{def}}{}}{=}\frac{\Gamma}{\Gamma\cdot [-(1-z)\cdot\ln(1-z)-z]+1}-\frac{\Gamma \cdot (1-p\cdot M_\Gamma(z))}{\Gamma\cdot [-(1-z)\cdot\ln(1-z)\cdot(1-p)-z]+1}$
for $z\in[0,\beta_\Gamma],\ p \in [0,1]$
\end{itemize}

\subsection{Properties of $\Gamma-weakly-adverse$ instances and items}
In this subsection, we list out some properties that a $\Gamma-weakly-adverse$ item must satisfy, where $\Gamma \in (\Gamma_{PT},1)$, where $\Gamma_{PT}\approx0.7251$ defined as in Theorem \ref{Condition for Peng and Tang to be valid} is the competitive-ratio achieved by Peng and Tang's \citep{PT} algorithm. We defer the proofs for these properties to Appendix \ref{Proof of properties}.
\bigbreak
\begin{restatable}{lemmma}{Properties}\label{Properties}
For an instance of distributions $\{D_i\}$ with $n$ items and some $\Gamma \in (\Gamma_{PT},1)$, if the following inequality holds for some $i\in[n]$
$$\hat{G}_i(0,\Gamma) >1 $$
then the following properties are satisfied by the function $\tilde{p}_i(x):$
\bigbreak
\noindent \textbf{Property A:} For any $x\in [0,\beta_\Gamma)$ such that $x<1-\frac{1}{e}$,
$$\int_0^x \tilde{p}_i(z)dz \leq \mu_\Gamma(x).$$
\bigbreak
\noindent \textbf{Property B:} For any $x \in (\gamma_\Gamma,1]$,
$$\int_{\gamma_\Gamma}^{x}Y_\Gamma(z,\tilde{p}_i(x))\cdot dz  < \Gamma\cdot (1-\beta_\Gamma)-H_\Gamma(\beta_\Gamma)\cdot(1-\Gamma\cdot\beta_\Gamma).$$
\bigbreak
\noindent \textbf{Property C:} For any $x \in [0,\beta_\Gamma)$,
$$\int_{x}^{\beta_\Gamma}W_\Gamma(z,\tilde{p}_i(x))\cdot dz  < H_\Gamma(\gamma_\Gamma)-M_\Gamma(\gamma_\Gamma).$$
\bigbreak
\end{restatable}
We shall also use that $Y_\Gamma(z,p)$ and $W_\Gamma(z,p)$ are non-negative in their respective domains, and that $Y_\Gamma(z,p)$ is non-increasing with respect to $p$ for a fixed $z \in [\gamma_\Gamma,1)$, while $W_\Gamma(z,p)$ is non-decreasing with respect to $p$ for a fixed $z \in [0,\beta_\Gamma]$.
The proofs for these claims are also contained in the proof of Lemma \ref{Properties} (in Appendix \ref{Proof of properties}).

\subsection{A bound on the weak-adverseness of the remaining items in an instance when one item is $\Gamma^*-weakly-adverse$}
In this subsection, we prove the following lemma. 

\begin{lemmma}\label{Simultaneous hardness}
If for an instance of distributions $\{D_i\}$ with $n$ items, there exists $i\in[n]$ such that item $i$ is $\Gamma^*-weakly-adverse$, then no other item in the instance is $\Gamma'-weakly-adverse$, where $\Gamma^*=0.7258$ and $\Gamma'=0.7276$. 
\end{lemmma}
 Before we jump into the proof for Lemma \ref{Simultaneous hardness}, we prove a useful result.
\begin{lemmma}\label{p_i(x)>y}
For $i\in[n]$, $\tilde{p}_i(x)>y$
holds for some $x,y \in [0,1]$ if and only if
$p_i(1-(1-x)\cdot(1-y)) > y$.
\begin{proof}
Assume that $p_i(x)>y$ for some $x,y\in[0,1]$. Since $1-t=(1-p_i(t))(1-q_i(t))$ for all $t \in [0,1]$, we have that 
$1-q_i^{-1}(x)=(1-\tilde{p}_i(x))\cdot(1-x)$.
Using $\tilde{p}_i(x)>y$, we get $q_i^{-1}(x) > 1-(1-x)\cdot (1-y)$.
Since $q_i(t)$ is increasing on $t\in[0,1]$, we now have
$$x > q_i(1-(1-x)\cdot(1-y)).$$
From $1-t=(1-p_i(t))(1-q_i(t))$, we also have
$$1-[1-(1-x)\cdot(1-y)]=(1-p_i(1-(1-x)\cdot(1-y)))(1-q_i(1-(1-x)\cdot(1-y))).$$
Since $q_i(1-(1-x)\cdot(1-y)) < x$, we obtain
$$p_i(1-(1-x)\cdot(1-y)) > y.$$
The proof for the converse can be obtained by inverting the signs of the inequalities in this proof.

\end{proof}
\end{lemmma}
 We now prove Lemma \ref{Simultaneous hardness}.
\begin{proof}
We have $\Gamma^*=0.7258$ and $\Gamma'=0.7276$. Numerically, we obtain the following values:
$$0.7879 \leq \beta_{\Gamma^*} \leq 0.7880 \textrm{\ \ \ \ \  \ \ \ } 0.7850 \leq \beta_{\Gamma'} \leq 0.7851$$
$$0.7893 \leq \gamma_{\Gamma^*} \leq 0.7894 \textrm{\ \ \ \ \  \ \ \ } 0.7900 \leq \gamma_{\Gamma'} \leq 0.7901$$
We also numerically obtain the following: 
$$H_{\Gamma^*}(\gamma_{\Gamma^*})-M_{\Gamma^*}(\gamma_{\Gamma^*}) < 0.00163 \ \ \ \ \ \ \ \ \ 
H_{\Gamma'}(\gamma_{\Gamma'})-M_{\Gamma'}(\gamma_{\Gamma'}) <0.00555 $$
$$\Gamma^*\cdot (1-\beta_{\Gamma^*})-H_{\Gamma^*}(\beta_{\Gamma^*})\cdot(1-\Gamma^*\cdot\beta_{\Gamma^*}) < 0.00068 \ \ \ \ \ \ \ \ 
\Gamma'\cdot (1-\beta_{\Gamma'})-H_{\Gamma'}(\beta_{\Gamma'})\cdot(1-\Gamma'\cdot\beta_{\Gamma'}) < 0.00237$$
Now suppose, for the sake of contradiction, that in an instance of distributions $\{D_i\}$ item $i$ is $\Gamma^*-weakly-adverse$ and item $j$ is $\Gamma'-weakly-adverse$. Applying \textbf{Property B} of Lemma \ref{Properties} on item $i$ with $x=0.9$, we get
$$\int_{\gamma_{\Gamma^*}}^{0.9} Y_{\Gamma^*}(z,\tilde{p}_i(0.9))\cdot dz< \Gamma^*\cdot (1-\beta_{\Gamma^*})-H_{\Gamma^*}(\beta_{\Gamma^*})\cdot(1-\Gamma^*\cdot\beta_{\Gamma^*})<0.00068.$$
Since $Y_{\Gamma^*}(z,p)$ is non-negative and $\gamma_{\Gamma^*} \leq 0.7894$, we have
$$\int_{0.7894}^{0.9} Y_{\Gamma^*}(z,\tilde{p}_i(0.9))\cdot dz<0.00068.$$
We now show that $\tilde{p}_i(0.9)>0.9$. We numerically obtain the following result,
$$\int_{0.7894}^{0.9} Y_{\Gamma^*}(z,0.9)\cdot dz\approx 0.000693>0.00068.$$
Since $Y_{\Gamma^*}(z,p)$ is non-increasing with respect to $p$ for a fixed $z \in [\gamma_\Gamma,1)$, we must have $\tilde{p}_i(0.9)>0.9$. We do a similar analysis for item $j$. Since $j$ is $\Gamma'-weakly-adverse$, we must have (putting $x=0.947$ in \textbf{Property B} of Lemma \ref{Properties})
$$\int_{\gamma_{\Gamma'}}^{0.947} Y_{\Gamma'}(z,\tilde{p}_j(0.947))\cdot dz< \Gamma'\cdot (1-\beta_{\Gamma'})-H_{\Gamma'}(\beta_{\Gamma'})\cdot(1-\Gamma'\cdot\beta_{\Gamma'}) < 0.00237.$$
Since $Y_{\Gamma'}(z,p)$ is non-negative and $\gamma_{\Gamma'} \leq 0.7901$, we must have
$$\int_{0.7901}^{0.947} Y_{\Gamma'}(z,\tilde{p}_j(0.947))\cdot dz<0.00237.$$
We numerically obtain
$$\int_{0.7901}^{0.947} Y_{\Gamma'}(z,0.811)\cdot dz\approx 0.002384 >0.00237.$$
From here, we conclude that $\tilde{p}_j(0.947)>0.811$.
We now do a similar analysis using the $W_\Gamma(z,p)$ function. Applying \textbf{Property C} of Lemma \ref{Properties} on item $i$, we obtain (putting $x=0.67$)
$$\int_{0.67}^{\beta_{\Gamma^*}} W_{\Gamma^*}(z,\tilde{p}_i(0.67))\cdot dz< H_{\Gamma^*}(\beta_{\Gamma^*})-M_{\Gamma^*}(\beta_{\Gamma^*})<0.00163.$$
Also, we numerically obtain the following result,
$$\int_{0.67}^{0.7879} W_{\Gamma^*}(z,0.2)\cdot dz \approx0.00165 >0.00163.$$
Since  $W_{\Gamma}(z,p)$ is non-decreasing with respect to $p$ at a fixed $z \in [0,\beta_\Gamma]$, we have that $\tilde{p}_i(0.67)<0.2$. We now have the following bounds:
$$\tilde{p}_i(0.9)>0.9 {\ \ \ \ \ } \tilde{p}_j(0.947)>0.811 {\ \ \ \ \ }\tilde{p}_i(0.67)<0.2 $$
We now use the above results to derive a contradiction. Directly applying Lemma \ref{p_i(x)>y} on the above inequalities, we obtain the following results:
$$p_i(0.99)>0.9 {\ \ \ \ \ } p_j(0.989983)>0.811 {\ \ \ \ \ } p_i(0.736)<0.2$$
Since $p_j(t)$ is an increasing function, we have:
$$p_j(0.99)>0.811.$$
Observe that $\frac{(1-p_i(t))\cdot(1-p_j(t))}{1-t}$ is a non-decreasing function of $t$ because 
$\frac{(1-p_i(t))\cdot(1-p_j(t))}{1-t}=\frac{1}{\mathbb{P}[\max_{k\neq i,j }v_k \leq \tau(t)]}$.
Hence, we have the following:
$$\frac{(1-p_i(0.99))\cdot(1-p_j(0.99))}{1-0.99}\geq \frac{(1-p_i(0.736))\cdot(1-p_j(0.736))}{1-0.736}.$$
Using the bounds computed above, we obtain from the above inequality:
$$p_j(0.736)>0.3763.$$
Note that the following statement follows from Lemma \ref{p_i(x)>y}:
$$p_i(t)>y\textrm{\ \ \ if and only if\ \ \ } \tilde{p}_i\left(1-\frac{1-t}{1-y}\right)>y$$
for $y,t \in [0,1)$ and $y < t$.
Applying this to $p_j(0.736)>0.3763$, we get that  $\tilde{p}_j(0.58)>0.3763$. Note that we have safely rounded up the argument of $\tilde{p}_j(x)$ in this obtained inequality, since $\tilde{p}_j(x)$ is non-decreasing.
Finally, we test \textbf{Property C} of Lemma \ref{Properties} for item $j$, by putting $x=0.58$. Since we have assumed item $j$ to be $\Gamma'-weakly-adverse$, we must have
$$\int_{0.58}^{\beta_{\Gamma'}}W_{\Gamma'}(z,\tilde{p}_i(0.58))\cdot dz < H_{\Gamma'}(\beta_{\Gamma'})-M_{\Gamma'}(\beta_{\Gamma'})< 0.00555.$$
Since $W_{\Gamma}(z,p)$ is non-negative and non-decreasing with respect to $p$, and $\beta_{\Gamma'}\geq 0.7850$, we have
$$\int_{0.58}^{0.7850}W_{\Gamma'}(z,0.3763)\cdot dz \leq \int_{0.58}^{\beta_{\Gamma'}}W_{\Gamma'}(z,\tilde{p}_i(0.58))\cdot dz < 0.00555.$$
The LHS of the above inequality evaluates to approximately $0.0096$, which is a contradiction. From here we conclude that there can not exist two items $i$ and $j$ in an instance such that item $i$ is $0.7258-weakly-adverse$ and item $j$ is $0.7276-weakly-adverse$.
\end{proof}

\subsection{Wrapping up}
We are now all set to prove Theorem \ref{Correctness of 2-scheme ALG}. Recall the statement of the theorem:
\correctness*
\begin{proof}
We begin with the assumption that item $1$ is $\Gamma^*-adverse$. Hence, it is also $\Gamma^*-weakly-adverse$ and satisfies the properties from Lemma \ref{Properties}. Since $c=0.28<\beta_{\Gamma^*}\approx0.7876$ and $c<1-\frac{1}{e}$, we can apply \textbf{Property A} to obtain
$$\int_0^x \tilde{p}_1(z)dz \leq \mu_{\Gamma^*}(x)$$
for all $x \in [0,c-\epsilon)$. From Lemma \ref{Bound on g1 tilde}, we now have for all $x \in [0,c)$,
$$\tilde{\bar{g}}_1(x,\Gamma^*)  \geq 1-\Gamma^*\cdot\int_0^x\tilde{p}_1(s)\cdot ds - \Gamma^* \cdot \epsilon\geq 1-\Gamma^*\cdot \mu_{\Gamma^*}(x)- \Gamma^* \cdot \epsilon$$
and for all $x \in [c,1)$,
$$\tilde{\bar{g}}_1(x,\Gamma^*) \geq \Gamma^*\cdot [-(1-x)\cdot \ln(1-x)\cdot (1-\tilde{p}_1(x))-x]+1. $$
\bigbreak
\begin{restatable}{lemmma}{Gmyu}\label{Bound on G as a function of myu}
Suppose that item 1 is $\Gamma-weakly-adverse$, and for some $\Gamma \in (0,1)$, $c\in[0,\beta_\Gamma)$ and small $\epsilon>0$, the functions $\tilde{\bar{g}},\eta:[0,1]\rightarrow[0,1]$ satisfy
$$\tilde{\bar{g}}(x)\geq \eta(x) > 0 $$
for $x\in[0,c-\epsilon)$ and 
$$\tilde{\bar{g}}(x) \geq \Gamma^\cdot [-(1-x)\cdot \ln(1-x)\cdot (1-\tilde{p}_1(x))-x]+1 $$
for $x\in[c,1]$, then
$$\int_0^1 \frac{\Gamma}{\tilde{\bar{g}}(x)\exp\left(\Gamma \cdot \int_0^x\frac{\tilde{p}_1(y)}{\tilde{\bar{g}}(y)}dy\right)}dx \leq K_\Gamma(\gamma_\Gamma)+M_\Gamma(c)-M_\Gamma(\gamma_\Gamma)+\int_0^{c} \frac{\Gamma}{\eta(x)}\cdot dx+O(\epsilon).$$

\end{restatable}
\begin{proof}
We defer the proof for this Lemma to Appendix \ref{Proof of bound on G as a function of myu}.
\end{proof}
\bigbreak
Applying Lemma \ref{Bound on G as a function of myu}
on the result obtained from Lemma \ref{Bound on g1 tilde}, by substituting $\Gamma=\Gamma^*=0.7258$, $c=0.28$, $\tilde{\bar{g}}(x)=\tilde{\bar{g}}_1(x, \Gamma)$, $\epsilon \rightarrow 0^+$ and $\eta(x)=1-\Gamma^*\cdot \mu_{\Gamma^*}(x)-\Gamma^*\cdot \epsilon$, we obtain
$$\int_0^1 \frac{\Gamma^*}{\tilde{\bar{g}}_1(x, \Gamma^*)\exp\left(\Gamma^* \cdot \int_0^x\frac{\tilde{p}_1(y)}{\tilde{\bar{g}}_1(y, \Gamma^*)}dy\right)}dx \leq K_{\Gamma^*}(\gamma_{\Gamma^*})+M_{\Gamma^*}(c)-M_{\Gamma^*}(\gamma_{\Gamma^*})+\int_0^c \frac{\Gamma^*}{1-\Gamma^*\cdot \mu_{\Gamma^*}(x)}\cdot dx.$$
For the chosen values, the RHS evaluates to approximately $0.9998 < 1$. Now, from (\ref{bar f_i(t) with changed variables}), we have
\begin{equation} \label{f1 bar final result}
\int_0^1 \bar{f}_1(t,\Gamma^*)\cdot dt =\int_0^1 \frac{\Gamma^*}{\tilde{\bar{g}}_1(x, \Gamma^*)\exp\left(\Gamma^* \cdot \int_0^x\frac{\tilde{p}_1(y)}{\tilde{\bar{g}}_1(y, \Gamma^*)}dy\right)}dx<1.
\end{equation}

\bigbreak
\begin{lemmma}\label{Proof for i neq 1}
If item $1$ is a $\Gamma^*-adverse$ item, then for all $i \neq 1$, $$\int_0^1 \bar{f}_i(t,\Gamma^*)\cdot dt \leq 1 $$
\end{lemmma}
\begin{proof}
From (\ref{bar f_i(t) with changed variables}), we have
$$\int_0^1 \bar{f}_i(t,\Gamma^*) =\int_0^1 \frac{\Gamma^*}{\tilde{\bar{g}}_i(x, \Gamma^*)\exp\left(\Gamma^* \cdot \int_0^x\frac{\tilde{p}_i(y)}{\tilde{\bar{g}}_i(y, \Gamma^*)}dy\right)}dx.$$
From Lemma \ref{Bound on gi tilde}, we  have for $i\neq 1$ and $x \in [0,1)$,
$$\tilde{\bar{g}}_i(x,\Gamma^*) \geq \Gamma^*\cdot [-(1-x)\cdot \ln(1-x)\cdot (1-\tilde{p}_i(x))-x]+1-\Gamma^*\cdot\mu_{\Gamma^*}(c),$$
which can be rewritten as
$$\tilde{\bar{g}}_i(x,\Gamma^*) \geq (1-\Gamma^*\cdot\mu_{\Gamma^*}(c)) \cdot \left[\frac{\Gamma^*}{(1-\Gamma^*\cdot\mu_{\Gamma^*}(c))}\cdot [-(1-x)\cdot \ln(1-x)\cdot (1-\tilde{p}_i(x))-x]+1\right].$$
Recall the definition of $\hat{g}_i(x, \Gamma)$.
$$\hat{g}_i(x, \Gamma) \overset{\underset{\mathrm{def}}{}}{=} \Gamma\cdot [-(1-x)\cdot \ln(1-x)\cdot (1-\tilde{p}_i(x))-x]+1.$$
Let $\Gamma'=\frac{\Gamma^*}{1-\Gamma^*\cdot\mu_{\Gamma^*}(c)} \approx 0.72759$. We have
$$\tilde{\bar{g}}_i(x,\Gamma^*) \geq \frac{\Gamma^*}{\Gamma'} \cdot \left(\Gamma'\cdot [-(1-x)\cdot \ln(1-x)\cdot (1-\tilde{p}_i(x))-x]+1\right) = \frac{\Gamma^*}{\Gamma'} \cdot \hat{g}_i(x, \Gamma').$$
Since we have $\tilde{\bar{g}}_i(x,\Gamma^*)\geq \frac{\Gamma^*}{\Gamma'} \cdot \hat{g}_i(x, \Gamma')$, we can apply Lemma \ref{rho1>rho2} to obtain
$$\int_0^1 \bar{f}_i(t,\Gamma^*) \leq1 \textrm{\ \ \ if \ \ \ }  \int_0^1 \frac{\Gamma^*}{\frac{\Gamma^*}{\Gamma'} \cdot \hat{g}_i(x, \Gamma')\exp\left(\Gamma^* \cdot \int_0^x\frac{\tilde{p}_i(y)}{\frac{\Gamma^*}{\Gamma'} \cdot \hat{g}_i(x, \Gamma')}dy\right)}dx \leq 1.$$
Simplifying, we get
$$\int_0^1 \bar{f}_i(t,\Gamma^*) \leq 1 \textrm{\ \ \ if \ \ \ } \int_0^1 \frac{\Gamma'}{ \hat{g}_i(x, \Gamma')\exp\left(\Gamma' \cdot \int_0^x\frac{\tilde{p}_i(y)}{  \hat{g}_i(x, \Gamma')}dy\right)}dx=\hat{G}_i(0,\Gamma')\leq 1.$$
Since item $1$ is 0.7258-weakly-adverse, from Lemma \ref{Simultaneous hardness} we have that item $i (i \neq 1)$ can not be 0.7276-weakly-adverse, and since $\Gamma'=\frac{\Gamma^*}{1-\Gamma^*\cdot\mu_{\Gamma^*}(c)} \approx 0.72759 < 0.7276$, we can apply Lemma \ref{Ordering on gamma} to conclude that item $i$ can not be $\Gamma'-weakly-adverse$. Hence, we have $G_i(0,\Gamma')
\leq 1$, and finally we have for $i \neq 1$
$$\int_0^1 \bar{f}_i(t,\Gamma^*)\cdot dt \leq 1.$$

\end{proof}
 From (\ref{f1 bar final result}) and Lemma \ref{Proof for i neq 1}, we have for all $i\in[n]$,
$$\int_0^1 \bar{f}_i(t,\Gamma^*) \cdot dt \leq 1.$$
 This completes our proof for Theorem \ref{Correctness of 2-scheme ALG}. 
\end{proof}

\section{Proof of Lemmas \ref{rho1>rho2} and \ref{Ordering on gamma}}\label{proof of rho1>rho2}
In this section, we prove Lemmas \ref{rho1>rho2} and \ref{Ordering on gamma}.
\rhoRho*
\begin{proof}
Define 
$$R_1(z) \overset{\underset{\mathrm{def}}{}}{=} \int_z^1 \frac{\Gamma}{{\rho}_1(x)\exp\left(\Gamma \cdot \int_z^x\frac{\tilde{p}(y)}{{\rho}_1(y)}dy\right)}dx {\ \ \ \ and \ \ \ \ }R_2(z) \overset{\underset{\mathrm{def}}{}}{=} \int_z^1 \frac{\Gamma}{{\rho}_2(x)\exp\left(\Gamma \cdot \int_z^x\frac{\tilde{p}(y)}{{\rho}_2(y)}dy\right)}dx.$$
Our aim is to show that $R_1(0)\leq 1$ if $R_2(0)\leq 1$.
Taking derivative, we obtain
$$R_1'(z) = \frac{\Gamma}{\rho_1(z)}\cdot(\tilde{p}(z)\cdot R_1(z)-1) {\ \ \ \ and \ \ \ \ }R_2'(z) = \frac{\Gamma}{\rho_2(z)}\cdot(\tilde{p}(z)\cdot R_2(z)-1).$$
From their definitions, we have $R_1(1)=R_2(1)=0$. Suppose we have $R_2(0)\leq 1$. We can show that $R_2(x)\leq 1$ for all $x\in[0,1]$ by contradiction. Suppose that $R_2(x_0)>1$ for some $x_0\in(0,1)$. There must be some $x_1\in[0,x_0)$ for which $R_2(x_1)=1$ and $R_2'(x_1)\geq0$. From the expression for $R_2'(x)$, it can be seen that this is not possible, since $\tilde{p}(x)-1 < 0$ for $x\in[0,1)$.

 We now have that $R_2(x)\leq 1$ for all $x\in[0,1]$ if $R_2(0)\leq 1$. We now show that $R_1(x)\leq R_2(x)$ for all $x\in[0,1]$ if $R_2(0)\leq 1$, by contradiction. Suppose that for some $x_0\in[0,1)$, $R_1(x_0)>R_2(x_0)$. Since $R_1(1)=R_2(1)=0$, there must be point $x_1\in[x_0,1]$ for which $R_1(x_1)>R_2(x_1)$ and $R_1'(x_1)<R_2'(x_1)$. From the expressions for $R_1'(x)$ and $R_2'(x)$, and using $\rho_1(x) \geq \rho_2(x)$ and $R_2(x)\leq 1$, it can be seen that this is not possible.

Hence, we have that $R_1(0)\leq R_2(0)$ if $R_2(0)\leq 1$.
This completes our proof for Lemma \ref{rho1>rho2}.
\bigbreak
\end{proof}
\ordering*
\begin{proof}
In order to prove Lemma \ref{Ordering on gamma}, we  show the contrapositive: For $0<\Gamma_1<\Gamma_2< 1$, $\hat{G}_i(0,\Gamma_1) \leq 1 $ if $\hat{G}_i(0,\Gamma_2) \leq 1 $.
We have that 

$$\hat{G}_i(z,\Gamma_1)=\int_z^1 \frac{\Gamma_1}{\hat{g}_i(x, \Gamma_1)\exp\left(\Gamma_1 \cdot \int_z^x\frac{\tilde{p}_i(y)}{\hat{g_i}(y, \Gamma_1)}dy\right)}dx {\ \ \ and \ \ \ } \hat{G}_i(z,\Gamma_2)=\int_z^1 \frac{\Gamma_2}{\hat{g}_i(x, \Gamma_2)\exp\left(\Gamma_2 \cdot \int_z^x\frac{\tilde{p}_i(y)}{\hat{g_i}(y, \Gamma_2)}dy\right)}dx.$$
We also have
$$\hat{g}_i(x, \Gamma_2) = \Gamma_2\cdot [-(1-x)\cdot \ln(1-x)\cdot (1-\tilde{p}_i(x))-x]+1 = \frac{\Gamma_2}{\Gamma_1} \cdot \left[\hat{g}_i(x, \Gamma_1)+\frac{\Gamma_1}{\Gamma_2}-1\right].$$
This gives us
$$\hat{G}_i(z,\Gamma_2)=\int_z^1 \frac{\Gamma_1}{\left[\hat{g}_i(x, \Gamma_1)+\frac{\Gamma_1}{\Gamma_2}-1\right]\exp\left(\Gamma_1 \cdot \int_z^x\frac{\tilde{p}_i(y)}{\left[\hat{g}_i(x, \Gamma_1)+\frac{\Gamma_1}{\Gamma_2}-1\right]}dy\right)}dx.$$
Since $\hat{g}_i(x, \Gamma_1)\geq \left[\hat{g}_i(x, \Gamma_1)+\frac{\Gamma_1}{\Gamma_2}-1\right] = \frac{\Gamma_1}{\Gamma_2}\cdot\hat{g}_i(x, \Gamma_2)>0$ for all $x\in[0,1]$, we can directly apply Lemma \ref{rho1>rho2} with $\rho_1(x)=\hat{g}_i(x, \Gamma_1)$,  $\rho_2(x)=\left[\hat{g}_i(x, \Gamma_1)+\frac{\Gamma_1}{\Gamma_2}-1\right]$, $\tilde{p}(x)=\tilde{p}_i(x)$ and $\Gamma=\Gamma_1$ to obtain
$$\hat{G}_i(0,\Gamma_1) \leq 1 \textrm{\ \ \  if\ \ \ } \hat{G}_i(0,\Gamma_2) \leq 1 .$$
This completes our proof for Lemma \ref{Ordering on gamma}.
\end{proof}
\section{Uniqueness of $\gamma_\Gamma$}\label{Uniqueness of gamma}
Recall that $\gamma_\Gamma$ is the root of $H_\Gamma(z)=K_\Gamma(z)$ on $z\in[0,1)$. This gives us 
$$\frac{\Gamma\cdot(-(1-\gamma_\Gamma)\cdot\ln(1-\gamma_\Gamma))}{\Gamma\cdot(-(1-\gamma_\Gamma)\cdot\ln(1-\gamma_\Gamma)-\gamma_\Gamma)+1}=\frac{\Gamma\cdot (1-\gamma_\Gamma)}{1-\Gamma\cdot\gamma_\Gamma}.$$
Simplifying, we get 
$$\Gamma=\frac{\ln(1-\gamma_\Gamma)+1}{\ln(1-\gamma_\Gamma)+\gamma_\Gamma}.$$
Observe that
$$\frac{d}{dx}\left(\frac{\ln(1-x)+1}{\ln(1-x)+x}\right) = \frac{-\ln(1-x)}{(\ln(1-x)+x)^2}>0$$ for $x\in[0,1)$. Also, $\frac{\ln(1-x)+1}{\ln(1-x)+x}$ equals 0 at $x=1-\frac{1}{e}$ and approaches 1 as $x$ approaches 1. Hence, there is a unique value $\gamma_\Gamma \in (1-\frac{1}{e},1)$ that satisfies $\Gamma=\frac{\ln(1-\gamma_\Gamma)+1}{\ln(1-\gamma_\Gamma)+\gamma_\Gamma}$. 
\section{Proof of Lemma \ref{Properties}: Properties of $\Gamma-weakly-adverse$ items }\label{Proof of properties}
Recall the statement of Lemma \ref{Properties}.
\Properties*
\begin{proof}
Parts of this proof have been borrowed from \citet{PT}.
\bigbreak
 Consider a $\Gamma-weakly-adverse$ item $i$ from an instance of distributions $\{D_i\}$, where $\Gamma>\Gamma_{PT}$. We have that $\hat{G}_i(0,\Gamma)>1$. For the ease of presentation, we will denote $\hat{G}_i(z,\Gamma)$ by $\hat{G}_\Gamma(z)$. We will denote the derivative of $\hat{G}_i(z,\Gamma)$ with respect to $z$ as $\hat{G}_\Gamma'(z)$. Taking the derivative of $\hat{G}_\Gamma(z)$ with respect to $z$, we get
$$\hat{G}_\Gamma'(z)=\frac{d}{dz}\int_z^1 \frac{\Gamma}{\hat{g}_i(x, \Gamma)\exp\left(\Gamma \cdot \int_z^x\frac{\tilde{p}_i(y)}{\hat{g}_i(y, \Gamma)}dy\right)}dx$$
$$=-\frac{\Gamma}{\hat{g}_i(z,\Gamma)}+\int_z^1 \frac{\Gamma}{\hat{g}_i(x, \Gamma)\exp\left(\Gamma \cdot \int_z^x\frac{\tilde{p}_i(y)}{\hat{g}_i(y, \Gamma)}dy\right)} \cdot \frac{\Gamma\cdot\tilde{p}_i(z)}{\hat{g}_i(z, \Gamma)}\cdot dx$$
$$=\frac{\Gamma}{\hat{g}_i(z,\Gamma)}\cdot(\tilde{p}_i(z)\cdot \hat{G}_\Gamma(z)-1) =\frac{\Gamma\cdot\left(\tilde{p}_i(z)\cdot \hat{G}_\Gamma(z)-1\right)}{\Gamma\cdot[-(1-z)\cdot\ln(1-z)\cdot(1-\tilde{p}_i(z))-z]+1} .$$
From the form of the expression above, it is evident that $\hat{G}_\Gamma'(z)$ is monotonic with respect to $\tilde{p}_i(z)$ at a fixed $z$. Hence, its minimum value must occur at either $\tilde{p}_i(z)=0$ or $\tilde{p}_i(z)=1$. We have 
$$\hat{G}_\Gamma'(z) \geq \min\left(\frac{-\Gamma}{\Gamma\cdot[-(1-z)\cdot\ln(1-z)-z]+1}, \frac{-\Gamma\cdot\left(1-\hat{G}_\Gamma(z)\right)}{1-\Gamma\cdot z}\right).$$
For $\hat{G}_\Gamma(z)\leq H_\Gamma(z)$, we have
$$\hat{G}_\Gamma'(z) \geq  \frac{-\Gamma\cdot\left(1-\hat{G}_\Gamma(z)\right)}{1-\Gamma\cdot z},$$ and for $\hat{G}_\Gamma(z)> H_\Gamma(z)$, we have $$\hat{G}_\Gamma'(z) \geq \frac{-\Gamma}{\Gamma\cdot[-(1-z)\cdot\ln(1-z)-z]+1}.$$

 From the analysis in Appendix \ref{Uniqueness of gamma} we have that $H_\Gamma(z) < K_\Gamma(z)$ for $z\in[0,\gamma_\Gamma)$ and $H_\Gamma(z) > K_\Gamma(z)$ for $z\in(\gamma_\Gamma,1)$. 
 Define $z_0\overset{\underset{\mathrm{def}}{}}{=}\inf\{z|\hat{G}_\Gamma(z)\leq H_\Gamma(z), z\in[0,1]\}$. $z_0$ is well defined because $G(1)=H(1)=1$. We now show that $z_0\leq \gamma_\Gamma$. It suffices to show that $\hat{G}_\Gamma(\gamma_\Gamma)\leq H_\Gamma(\gamma_\Gamma)$, which we now show by contradiction. Suppose, for contradiction, that  $\hat{G}_\Gamma(\gamma_\Gamma) > H_\Gamma(\gamma_\Gamma)$. For sufficiently small $\epsilon>0$, it holds that $\hat{G}_\Gamma(\gamma_\Gamma) < H_\Gamma(\gamma_\Gamma)$ for $z\in[1-\epsilon,1)$. This is because for $z\in[1-\epsilon,1)$,
$$\hat{G}'_\Gamma(z)-H_\Gamma'(z)=\hat{G}'_\Gamma(z)-\frac{\Gamma\cdot ((1-\Gamma)\cdot\ln(1-z)-\Gamma\cdot z+1)}{(\Gamma\cdot[-(1-z)\cdot\ln(1-z)-z]+1)^2}$$
$$\geq -\frac{\Gamma}{1-\Gamma\cdot z}-\frac{\Gamma\cdot ((1-\Gamma)\cdot\ln(1-z)-\Gamma\cdot z+1)}{(\Gamma\cdot[-(1-z)\cdot\ln(1-z)-z]+1)^2} >0.$$
The last inequality holds since $$\lim_{x\rightarrow 1^-} \left( -\frac{\Gamma}{1-\Gamma\cdot z}-\frac{\Gamma\cdot ((1-\Gamma)\cdot\ln(1-z)-\Gamma\cdot z+1)}{(\Gamma\cdot[-(1-z)\cdot\ln(1-z)-z]+1)^2} \right) = +\infty.$$
For $z\in[1-\epsilon,1)$, we have
$$\hat{G}_\Gamma(z)-H_\Gamma(z) = G(1)-H(1)-\int_z^1 (\hat{G}'_\Gamma(t)-H'_\Gamma(t))\cdot dt=-\int_z^1 (\hat{G}'_\Gamma(t)-H'_\Gamma(t))\cdot dt<0.$$
Let $z_2\overset{\underset{\mathrm{def}}{}}{=}\inf\{z|z\in(0,1], \hat{G}_\Gamma(y)\leq H_\Gamma(y) \textrm{\ for\ all \ } y \in [z,1) $. From the derivation above, we have $z_2\leq 1-\epsilon$. We also have $\hat{G}_\Gamma(z_2)=H_\Gamma(z_2)$, since both $\hat{G}_\Gamma$ and $H_\Gamma$ are continuous functions. Since $\hat{G}_\Gamma(\gamma_\Gamma)>H_\Gamma(\gamma_\Gamma)$, we have $z_2\in(\gamma_\Gamma, 1-\epsilon]$. According to the definition of $z_2$ and $H_\Gamma(z)$, for every $z \in [z_2, 1]$, we have $\hat{G}_\Gamma(z)\leq\hat{H}_\Gamma(z)$, and hence:
$$\hat{G}'_\Gamma(z) \geq \frac{-\Gamma\cdot(1-\Gamma(z))}{1-\Gamma\cdot z}$$
$$\implies ((1-\Gamma z)\cdot\hat{G}_\Gamma(z) )'= (1-\Gamma z)\cdot \hat{G}'_\Gamma(z) - \Gamma\cdot \hat{G}_\Gamma(z) \geq - \Gamma$$
$$\implies(1-\Gamma)\cdot\hat{G}_\Gamma(1) - (1-\Gamma z)\cdot \hat{G}_\Gamma(z) \geq -\Gamma\cdot(1-z)$$
$$\implies\hat{G}_\Gamma(z)\leq \frac{\Gamma\cdot (1-z)}{1-\Gamma z}=K(z).$$
We already have that $K_\Gamma(z)<H_\Gamma(z)$ for $z\in(\gamma_\Gamma,1)$. This implies that $\hat{G}_\Gamma(z)<H_\Gamma(z)$ for $z\in[z_2,1)$, which contradicts $\hat{G}_\Gamma(z_2)=H_\Gamma(z_2)$. Therefore, we have $z_0\leq\gamma_\Gamma$. In fact, since the contradiction arises only from the assertion that $z_2>\gamma_\Gamma$, we have that $z_2\leq\gamma_\Gamma$, and $\hat{G}_\Gamma(z)\leq K_\Gamma(z)$ holds for $z\in[z_2,1]$.
\bigbreak
 Taking the derivative of $H_\Gamma(z)$, we get
$$H'_\Gamma(z)=\Gamma\cdot\frac{1+\ln(1-z)-\Gamma\cdot(z+\ln(1-z))}{\left(\Gamma\cdot\left[-(1-z)\cdot\ln(1-z)-z\right]+1\right)^2}.$$
It can be verified from the definition of $\gamma_\Gamma$ that $H'_\Gamma(z)>0$ for $z\in[0,\gamma_\Gamma)$ and $H'_\Gamma(z)<0$ for $z\in(\gamma_\Gamma,1)$. 
\linebreak
 We now show that $\beta_\Gamma < z_0$. Recall that $\beta_\Gamma\overset{\underset{\mathrm{def}}{}}{=}\inf\{z|M_\Gamma(z)\leq H_\Gamma(z), z\in[0,1)\}$. Since we have that $\hat{G}_\Gamma(0)>1$, and since $\hat{G}_\Gamma(z)\geq H_\Gamma(z)$ for $z \in [0,z_0]$, we have for all $z\in[0,z_0],$
$$\hat{G}'_\Gamma(z) \geq \frac{-\Gamma}{\Gamma\cdot[-(1-z)\cdot\ln(1-z)-z]+1} $$
$$\implies \hat{G}_\Gamma(z) > 1+ \int_0^z \frac{-\Gamma}{\Gamma\cdot[-(1-x)\cdot\ln(1-x)-x]+1}\cdot dx = M_\Gamma(z) .$$
Hence, we have that $\hat{G}_\Gamma(z_0)> M_\Gamma(z_0)$. From the definition of $z_0$, we have $\hat{G}_\Gamma(z_0)=H_\Gamma(z_0)$. This implies that $H_\Gamma(z_0)> M_\Gamma(z_0)$. From the definition of $\beta_\Gamma$ it is now clear that $\beta_\Gamma< z_0$.
\bigbreak
 Before we go ahead and prove the properties from Lemma \ref{Properties}, we make the following observation. Here $V_\Gamma(z)$ is some function of $z$ parametrized by $\Gamma$.
$$\frac{\partial}{\partial p}\left[\frac{\Gamma\cdot(1-p\cdot V_\Gamma(z))}{\Gamma\cdot \left[-(1-z)\cdot\ln(1-z)\cdot(1-p)-z\right]+1}\right] = \frac{\Gamma\cdot \left(\Gamma\cdot \left[-(1-z)\cdot\ln(1-z)-z\right]+1\right)\cdot(H_\Gamma(z)-V_\Gamma(z))}{\left(\Gamma\cdot \left[-(1-z)\cdot\ln(1-z)\cdot(1-p)-z\right]+1\right)^2} $$
Substituting $V_\Gamma(z)=M_\Gamma(z)$ in the equation above and using the fact that $H_\Gamma(z) \leq M_\Gamma(z)$ on $z\in[0,\beta_\Gamma]$, we have that $$\frac{\partial}{\partial p}W_\Gamma(z,p) = - \frac{\partial}{\partial p}\left[\frac{\Gamma\cdot(1-p\cdot M_\Gamma(z))}{\Gamma\cdot \left[-(1-z)\cdot\ln(1-z)\cdot(1-p)-z\right]+1}\right]\geq 0$$ on $z\in[0,\beta_\Gamma]$. Similarly, substituting $V_\Gamma(z)=K_\Gamma(z)$ and using the fact that $H_\Gamma(z)\geq K_\Gamma(z)$ on $z\in[\gamma_\Gamma,1)$, we have that $$\frac{\partial}{\partial p}Y_\Gamma(z,p) = -(1-\Gamma z)\cdot \frac{\partial}{\partial p}\left[\frac{\Gamma\cdot(1-p\cdot K_\Gamma(z))}{\Gamma\cdot \left[-(1-z)\cdot\ln(1-z)\cdot(1-p)-z\right]+1}\right] \leq 0$$ on $z\in[\gamma_\Gamma,1)$. Since $W_\Gamma(z,0)=0$ and $Y_\Gamma(z,1)=0$, we also have that $W_\Gamma(z,p)$ is non-negative on $z\in[0,\beta_\Gamma]$, $p\in[0,1]$ and $Y_\Gamma(z,p)$ is non-negative on $z\in[\gamma_\Gamma,1)$, $p\in[0,1]$.
\bigbreak
 We now use $\hat{G}_\Gamma(0)>1$ and  $\hat{G}_\Gamma(z_0)=H_\Gamma(z_0)$.
$$\hat{G}_\Gamma(z_0)=\hat{G}_\Gamma(0)+\int_0^{z_0} \hat{G}'_\Gamma(z)\cdot dz$$
$$\implies H_\Gamma(z_0)>1+\int_0^{z_0} \hat{G}'_\Gamma(z)\cdot dz=1+\int_0^{z_0} \frac{\Gamma\cdot\left(\tilde{p}_i(z)\cdot \hat{G}_\Gamma(z)-1\right)}{\Gamma\cdot[-(1-z)\cdot\ln(1-z)\cdot(1-\tilde{p}_i(z))-z]+1} dz$$
$$\geq1+ \int_{\beta_\Gamma}^{z_0}-\frac{\Gamma}{\Gamma\cdot[-(1-z)\cdot\ln(1-z)-z]+1} dz+\int_0^{\beta_\Gamma} \frac{\Gamma\cdot\left(\tilde{p}_i(z)\cdot \hat{G}_\Gamma(z)-1\right)}{\Gamma\cdot[-(1-z)\cdot\ln(1-z)\cdot(1-\tilde{p}_i(z))-z]+1} dz $$
$$\geq1+ \int_{\beta_\Gamma}^{z_0}-\frac{\Gamma}{\Gamma\cdot[-(1-z)\cdot\ln(1-z)-z]+1} dz+\int_0^{\beta_\Gamma} \frac{\Gamma\cdot\left(\tilde{p}_i(z)\cdot M_\Gamma(z)-1\right)}{\Gamma\cdot[-(1-z)\cdot\ln(1-z)\cdot(1-\tilde{p}_i(z))-z]+1} dz .$$
The final inequality follows from the fact that $\hat{G}_\Gamma(z)>M_\Gamma(z)$ on $[0,a_\Gamma]$. We now have,
$$\int_0^{\beta_\Gamma} \frac{\Gamma\cdot\left(\tilde{p}_i(z)\cdot M_\Gamma(z)-1\right)}{\Gamma\cdot[-(1-z)\cdot\ln(1-z)\cdot(1-\tilde{p}_i(z))-z]+1} dz < H_\Gamma(z_0)+\int_{\beta_\Gamma}^{z_0}\frac{\Gamma}{\Gamma\cdot[-(1-z)\cdot\ln(1-z)-z]+1} dz-1$$
$$\implies \int_0^{\beta_\Gamma} \frac{\Gamma\cdot\left(\tilde{p}_i(z)\cdot M_\Gamma(z)-1\right)}{\Gamma\cdot[-(1-z)\cdot\ln(1-z)\cdot(1-\tilde{p}_i(z))-z]+1} dz < H_\Gamma(z_0)-M_\Gamma(z_0)+M_\Gamma(\beta_\Gamma)-1.$$
$H_\Gamma(z)$ is an increasing function on $[0,\gamma_\Gamma)$ and $M_\Gamma(z)$ is a decreasing function on $[0,1]$. Since $z_0 \leq \gamma_\Gamma$, we now have 
$$\int_0^{\beta_\Gamma} \frac{\Gamma\cdot\left(\tilde{p}_i(z)\cdot M_\Gamma(z)-1\right)}{\Gamma\cdot[-(1-z)\cdot\ln(1-z)\cdot(1-\tilde{p}_i(z))-z]+1} dz < H_\Gamma(\gamma_\Gamma)-M_\Gamma(\gamma_\Gamma)+M_\Gamma(\beta_\Gamma)-1$$
$$= H_\Gamma(\gamma_\Gamma)-M_\Gamma(\gamma_\Gamma)+\int_0^{\beta_\Gamma}-\frac{\Gamma}{\Gamma\cdot[-(1-z)\cdot\ln(1-z)-z]+1}dz$$
\begin{equation}\label{Pre property C}
\implies \int_0^{\beta_\Gamma} W_\Gamma(z,\tilde{p}_i(z))\cdot dz < H_\Gamma(\gamma_\Gamma)-M_\Gamma(\gamma_\Gamma)
\end{equation}

 Since $\frac{\partial}{\partial p}W_\Gamma(z,p) \geq 0$, $W_\Gamma(z,p) \geq 0$ on $z \in [0,\beta_\Gamma]$ and $\tilde{p}_i(z)$ is an increasing function, \textbf{Property C} follows from the above inequality, and we also obtain $W_\Gamma(z,p_0) \geq p_0\cdot \left[\frac{\partial}{\partial p}W_\Gamma(z,p)\right]_{p=0}$ for $p_0 \in [0,1]$. Hence, for $x\in [0,\beta_\Gamma]$,
$$\int_0^x W_\Gamma(z,\tilde{p}_i(z))\cdot dz \geq \int_0^x \tilde{p}_i(z) \cdot\left[\frac{\partial}{\partial p}W_\Gamma(z,p)\right]_{p=0} \cdot dz$$
$$= \frac{\Gamma}{\Gamma\cdot[-(1-z)\cdot\ln(1-z)-z]+1}\cdot(M_\Gamma(z)-H_\Gamma(z))\cdot \int_0^x \tilde{p}_i(z)\cdot dz.$$
Since $H_\Gamma(z)$ is non-decreasing on $z \in [0,\gamma_\Gamma]$ and $x\leq\beta_\Gamma\leq\gamma_\Gamma$, it is clear that $\frac{-H_\Gamma(z)}{\Gamma\cdot[-(1-z)\cdot\ln(1-z)-z]+1}$ is non-increasing on $[0,x]$. Also $\frac{d}{dz}\frac{M_\Gamma(z)}{\Gamma\cdot[-(1-z)\cdot\ln(1-z)-z]+1}=\frac{-\Gamma\cdot(1-M_\Gamma(z)\cdot\ln(1-z))}{\left(\Gamma\cdot[-(1-z)\cdot\ln(1-z)-z]+1\right)^2}$. Since $M_\Gamma(z)\leq 1$ (from definition) and assuming $\ln(1-z)\leq1$ (by taking $x\leq1-\frac{1}{e}$), we have that $\frac{M_\Gamma(z)}{\Gamma\cdot[-(1-z)\cdot\ln(1-z)-z]+1}$ is non-increasing on $z\in[0,x]$. Hence, we have for $z\in[0,x]$,
$$\frac{\Gamma}{\Gamma\cdot[-(1-z)\cdot\ln(1-z)-z]+1}\cdot(M_\Gamma(z)-H_\Gamma(z)) \geq \frac{\Gamma}{\Gamma\cdot[-(1-x)\cdot\ln(1-x)-x]+1}\cdot(M_\Gamma(x)-H_\Gamma(x)) .$$
This gives us
$$\int_0^x W_\Gamma(z,\tilde{p}_i(z))\cdot dz \geq  \frac{\Gamma}{\Gamma\cdot[-(1-x)\cdot\ln(1-x)-x]+1}\cdot(M_\Gamma(x)-H_\Gamma(x))\cdot \int_0^x \tilde{p}_i(z)\cdot dz$$
Since $x\leq\beta_\Gamma$, we have 
$$\int_0^{\beta_\Gamma} W_\Gamma(z,\tilde{p}_i(z))\cdot dz \geq  \frac{\Gamma}{\Gamma\cdot[-(1-x)\cdot\ln(1-x)-x]+1}\cdot(M_\Gamma(x)-H_\Gamma(x))\cdot \int_0^x \tilde{p}_i(z)\cdot dz.$$
Finally, from (\ref{Pre property C})
$$H_\Gamma(\gamma_\Gamma) - M_\Gamma(\gamma_\Gamma) \geq  \frac{\Gamma}{\Gamma\cdot[-(1-x)\cdot\ln(1-x)-x]+1}\cdot(M_\Gamma(x)-H_\Gamma(x))\cdot \int_0^x \tilde{p}_i(z)\cdot dz.$$
This completes our proof for \textbf{Property A}. 
\bigbreak
 Recall that $z_2\overset{\underset{\mathrm{def}}{}}{=}\inf\{z|z\in(0,1], \hat{G}_\Gamma(y)\leq H_\Gamma(y) \textrm{\ for\ all \ } y \in [z,1) \}$, and that $\hat{G}_\Gamma(z) \leq K_\Gamma(z)$ on $z\in[z_2,1]$. Since $\hat{G}_\Gamma(z) \leq H_\Gamma(z)$ on $z\in[z_2,1]$ (from the definition of $z_2$), we have 
$$\hat{G}'_\Gamma(z) = \frac{\Gamma\cdot\left(\tilde{p}_i(z)\cdot \hat{G}_\Gamma(z)-1\right)}{\Gamma\cdot[-(1-z)\cdot\ln(1-z)\cdot(1-\tilde{p}_i(z))-z]+1} \geq \frac{-\Gamma\cdot(1-\hat{G}_\Gamma(z))}{1-\Gamma z} + \frac{Y_\Gamma(z,\tilde{p}_i(z))}{1-\Gamma z}.$$
Since $Y_\Gamma(z,p)$ is only defined on $z \in [\gamma_\Gamma,1)$, we extend it to $[z_2,1]$ by assuming $Y_\Gamma(z,p)=0$ for $z\in[z_2,\gamma_\Gamma)$ in order to make the above inequality well defined and correct. We will be using this extended form of $Y_\Gamma(z,p)$ in the upcoming steps as well. We now have for $z\in[z_2,1]$,
$$\hat{G}'_\Gamma(z)\cdot(1-\Gamma z)-\Gamma\cdot\hat{G}_\Gamma(z) \geq   -\Gamma+Y_\Gamma(z,\tilde{p}_i(z))$$
$$\implies (\hat{G}_\Gamma(z)\cdot(1-\Gamma z))' \geq   -\Gamma+Y_\Gamma(z,\tilde{p}_i(z)).$$
On integrating both sides from $z_2$ to $1$, we obtain
$$-\hat{G}_\Gamma(z_2)\cdot(1-\Gamma z_2)\geq   -\Gamma\cdot(1-z_2)+\int_{z_2}^1 Y_\Gamma(z,\tilde{p}_i(z))\cdot dz.$$
Since $Y_\Gamma(z,\tilde{p}_i(z))$ is 0 on $z\in[z_2,\gamma_\Gamma)$, we have
$$\int_{\gamma_\Gamma}^1Y_\Gamma(z,\tilde{p}_i(z))\cdot dz \leq -\hat{G}_\Gamma(z_2)\cdot(1-\Gamma z_2)+\Gamma\cdot(1-z_2) .$$
$$\implies \int_{\gamma_\Gamma}^1Y_\Gamma(z,\tilde{p}_i(z))\cdot dz \leq -{H}_\Gamma(z_2)\cdot(1-\Gamma z_2)+\Gamma\cdot(1-z_2). $$
Notice that $\frac{d}{dz}\left[-{H}_\Gamma(z)\cdot(1-\Gamma z)+\Gamma\cdot(1-z)\right] = \Gamma\cdot(H_\Gamma(z)-1)-H_\Gamma'(z)\cdot(1-\Gamma z)<0$ on $z \in [0,\gamma_\Gamma)$ since $H_\Gamma'(z)>0$ on this interval and $H_\Gamma(z)<1$ always holds. Since $z_2\geq z_0$ (from their definitions) and $\beta_\Gamma < z_0$, we have $z_2>\beta_\Gamma$. This gives us
$$\int_{\gamma_\Gamma}^1Y_\Gamma(z,\tilde{p}_i(z))\cdot dz \leq -{H}_\Gamma(\beta_\Gamma)\cdot(1-\Gamma\cdot \beta_\Gamma)+\Gamma\cdot(1-\beta_\Gamma) .$$
Since $\frac{\partial}{\partial p}Y_\Gamma(z,p) \leq 0$, $Y_\Gamma(z,p) \geq 0$ on $z \in [\gamma_\Gamma,1]$ and $\tilde{p}_i(z)$ is an increasing function, \textbf{Property B} follows from the above inequality.
\end{proof}

\section{Proof of Lemma \ref{Bound on G as a function of myu} }\label{Proof of bound on G as a function of myu}
Recall the statement of Lemma \ref{Bound on G as a function of myu}.
\Gmyu*
\begin{proof}
Define
$$\bar{G}_\Gamma(z)\overset{\underset{\mathrm{def}}{}}{=}\int_z^1 \frac{\Gamma}{\tilde{\bar{g}}(x)\exp\left(\Gamma \cdot \int_z^x\frac{\tilde{p}_1(y)}{\tilde{\bar{g}}(y)}dy\right)}dx.$$
Performing the same analysis as done in Appendix \ref{Proof of properties}, we get for $z\in[c,1]$
$$\bar{G}'_\Gamma(z)=\frac{\Gamma\cdot\left(\tilde{p}_1(z)\cdot \bar{G}_\Gamma(z)-1\right)}{\tilde{\bar{g}}(z)} \geq \min\left(\frac{-\Gamma}{\Gamma\cdot[-(1-z)\cdot\ln(1-z)-z]+1}, \frac{-\Gamma\cdot\left(1-\bar{G}_\Gamma(z)\right)}{1-\Gamma\cdot z}\right).$$
From the analysis in Appendix \ref{Proof of properties}, we know that this is sufficient for us to show that there exists a $z_0 \in [\beta_\Gamma,\gamma_\Gamma]$ such that $\bar{G}_\Gamma(z_0)=H_\Gamma(z_0)$ and $z_0=\inf\{z|\bar{G}_\Gamma(z)\leq H_\Gamma(z), z\in[0,1]\}$. We refrain from rewriting the proofs of these claims for conciseness. Now, we have
\begin{equation}\label{Integral bound}
\bar{G}_\Gamma(0)=-\int_0^{z_0}\bar{G}'_\Gamma(z)\cdot dz + G_\Gamma(z_0) \leq -\int_0^{c-\epsilon}\bar{G}'_\Gamma(z)\cdot dz-\int_c^{z_0}\bar{G}'_\Gamma(z)\cdot dz + H_\Gamma(z_0) + O(\epsilon).
\end{equation}
Since $\bar{G}_\Gamma(z)\geq H_\Gamma(z)$ for $z<z_0$, we have that $$\int_c^{z_0}\bar{G}'_\Gamma(z)\cdot dz \geq \int_c^{z_0} \frac{-\Gamma}{\Gamma\cdot[-(1-z)\cdot\ln(1-z)-z]+1} \cdot dz = M_\Gamma(z_0)-M_\Gamma(c).$$
Also, since $\tilde{\bar{g}}(z) \geq \eta(z)>0$ for $z \in [0,c-\epsilon)$, we have
$\bar{G}_\Gamma'(z)=\frac{\Gamma\cdot\left(\tilde{p}_1(z)\cdot \bar{G}_\Gamma(z)-1\right)}{\tilde{\bar{g}}(z)} \geq -\frac{\Gamma}{\eta(z)}$.
Now, from (\ref{Integral bound}), we have
$$\bar{G}_\Gamma(0) \leq \int_0^{c-\epsilon}\frac{\Gamma}{\eta(z)}\cdot dz-M_\Gamma(z_0)+M_\Gamma(c) + H_\Gamma(z_0)+ O(\epsilon).$$
Note that $z_0\leq \gamma_\Gamma$ and $H_\Gamma(z)$ is increasing for $z<\gamma_\Gamma$, while $M_\Gamma(z)$ is a decreasing function. We have 
$$\bar{G}_\Gamma(0) \leq \int_0^{c}\frac{\Gamma}{\eta(z)}\cdot dz-M_\Gamma(\gamma_\Gamma)+M_\Gamma(c) + H_\Gamma(\gamma_\Gamma)+ O(\epsilon).$$
From the definition of $\gamma_\Gamma$, we have $H_\Gamma(\gamma_\Gamma)=K_\Gamma(\gamma_\Gamma)$. This gives us
$$\bar{G}_\Gamma(0) \leq \int_0^{c}\frac{\Gamma}{\eta(z)}\cdot dz-M_\Gamma(\gamma_\Gamma)+M_\Gamma(c) + K_\Gamma(\gamma_\Gamma)+ O(\epsilon).$$
This completes our proof for the lemma.
\end{proof}
\bibliographystyle{ACM-Reference-Format}
\bibliography{bibfile}


\begin{thebibliography}{25}


\ifx \showCODEN    \undefined \def \showCODEN     #1{\unskip}     \fi
\ifx \showDOI      \undefined \def \showDOI       #1{#1}\fi
\ifx \showISBNx    \undefined \def \showISBNx     #1{\unskip}     \fi
\ifx \showISBNxiii \undefined \def \showISBNxiii  #1{\unskip}     \fi
\ifx \showISSN     \undefined \def \showISSN      #1{\unskip}     \fi
\ifx \showLCCN     \undefined \def \showLCCN      #1{\unskip}     \fi
\ifx \shownote     \undefined \def \shownote      #1{#1}          \fi
\ifx \showarticletitle \undefined \def \showarticletitle #1{#1}   \fi
\ifx \showURL      \undefined \def \showURL       {\relax}        \fi
\providecommand\bibfield[2]{#2}
\providecommand\bibinfo[2]{#2}
\providecommand\natexlab[1]{#1}
\providecommand\showeprint[2][]{arXiv:#2}

\bibitem[Abolhassani et~al\mbox{.}(2017)]%
        {Abolhassani}
\bibfield{author}{\bibinfo{person}{Melika Abolhassani}, \bibinfo{person}{Soheil
  Ehsani}, \bibinfo{person}{Hossein Esfandiari}, \bibinfo{person}{MohammadTaghi
  Hajiaghayi}, \bibinfo{person}{Robert~D. Kleinberg}, {and}
  \bibinfo{person}{Brendan Lucier}.} \bibinfo{year}{2017}\natexlab{}.
\newblock \showarticletitle{Beating 1-1/e for ordered prophets}. In
  \bibinfo{booktitle}{\emph{{STOC}}}. \bibinfo{publisher}{{ACM}},
  \bibinfo{pages}{61--71}.
\newblock
\urldef\tempurl%
\url{https://doi.org/10.1145/3055399.3055479}
\showDOI{\tempurl}


\bibitem[Agrawal et~al\mbox{.}(2020)]%
        {Agrawal}
\bibfield{author}{\bibinfo{person}{Shipra Agrawal}, \bibinfo{person}{Jay
  Sethuraman}, {and} \bibinfo{person}{Xingyu Zhang}.}
  \bibinfo{year}{2020}\natexlab{}.
\newblock \showarticletitle{On Optimal Ordering in the Optimal Stopping
  Problem}. In \bibinfo{booktitle}{\emph{{EC}}}. \bibinfo{publisher}{{ACM}},
  \bibinfo{pages}{187--188}.
\newblock
\urldef\tempurl%
\url{https://doi.org/10.1145/3391403.3399484}
\showDOI{\tempurl}


\bibitem[Alaei(2011)]%
        {Alaei}
\bibfield{author}{\bibinfo{person}{Saeed Alaei}.}
  \bibinfo{year}{2011}\natexlab{}.
\newblock \showarticletitle{Bayesian Combinatorial Auctions: Expanding Single
  Buyer Mechanisms to Many Buyers}. In \bibinfo{booktitle}{\emph{{FOCS}}}.
  \bibinfo{publisher}{{IEEE} Computer Society}, \bibinfo{pages}{512--521}.
\newblock
\urldef\tempurl%
\url{https://doi.org/10.1109/FOCS.2011.90}
\showDOI{\tempurl}


\bibitem[Allaart(2007)]%
        {Allaart}
\bibfield{author}{\bibinfo{person}{Pieter~C. Allaart}.}
  \bibinfo{year}{2007}\natexlab{}.
\newblock \showarticletitle{Prophet Inequalities for I.I.D. Random Variables
  with Random Arrival Times}.
\newblock \bibinfo{journal}{\emph{Sequential Analysis}} \bibinfo{volume}{26},
  \bibinfo{number}{4} (\bibinfo{year}{2007}), \bibinfo{pages}{403--413}.
\newblock
\urldef\tempurl%
\url{https://doi.org/10.1080/07474940701620857}
\showDOI{\tempurl}
\showeprint{https://doi.org/10.1080/07474940701620857}


\bibitem[Arsenis et~al\mbox{.}(2021)]%
        {constrained}
\bibfield{author}{\bibinfo{person}{Makis Arsenis}, \bibinfo{person}{Odysseas
  Drosis}, {and} \bibinfo{person}{Robert Kleinberg}.}
  \bibinfo{year}{2021}\natexlab{}.
\newblock \showarticletitle{Constrained-Order Prophet Inequalities}. In
  \bibinfo{booktitle}{\emph{{SODA}}},
  \bibfield{editor}{\bibinfo{person}{D{\'{a}}niel Marx}} (Ed.).
  \bibinfo{publisher}{{SIAM}}, \bibinfo{pages}{2034--2046}.
\newblock
\urldef\tempurl%
\url{https://doi.org/10.1137/1.9781611976465.121}
\showDOI{\tempurl}


\bibitem[Azar et~al\mbox{.}(2018)]%
        {Azar}
\bibfield{author}{\bibinfo{person}{Yossi Azar}, \bibinfo{person}{Ashish
  Chiplunkar}, {and} \bibinfo{person}{Haim Kaplan}.}
  \bibinfo{year}{2018}\natexlab{}.
\newblock \showarticletitle{Prophet Secretary: Surpassing the 1-1/e Barrier}.
  In \bibinfo{booktitle}{\emph{{EC}}}. \bibinfo{publisher}{{ACM}},
  \bibinfo{pages}{303--318}.
\newblock
\urldef\tempurl%
\url{https://doi.org/10.1145/3219166.3219182}
\showDOI{\tempurl}


\bibitem[Beyhaghi et~al\mbox{.}(2021)]%
        {Beyhaghi}
\bibfield{author}{\bibinfo{person}{Hedyeh Beyhaghi}, \bibinfo{person}{Negin
  Golrezaei}, \bibinfo{person}{Renato~Paes Leme}, \bibinfo{person}{Martin
  P{\'{a}}l}, {and} \bibinfo{person}{Balasubramanian Sivan}.}
  \bibinfo{year}{2021}\natexlab{}.
\newblock \showarticletitle{Improved Revenue Bounds for Posted-Price and
  Second-Price Mechanisms}.
\newblock \bibinfo{journal}{\emph{Oper. Res.}} \bibinfo{volume}{69},
  \bibinfo{number}{6} (\bibinfo{year}{2021}), \bibinfo{pages}{1805--1822}.
\newblock
\urldef\tempurl%
\url{https://doi.org/10.1287/opre.2021.2121}
\showDOI{\tempurl}


\bibitem[Chakraborty et~al\mbox{.}(2010)]%
        {Chakraborty}
\bibfield{author}{\bibinfo{person}{Tanmoy Chakraborty}, \bibinfo{person}{Eyal
  Even{-}Dar}, \bibinfo{person}{Sudipto Guha}, \bibinfo{person}{Yishay
  Mansour}, {and} \bibinfo{person}{S. Muthukrishnan}.}
  \bibinfo{year}{2010}\natexlab{}.
\newblock \showarticletitle{Approximation Schemes for Sequential Posted Pricing
  in Multi-unit Auctions}. In \bibinfo{booktitle}{\emph{{WINE}}},
  Vol.~\bibinfo{volume}{6484}. \bibinfo{publisher}{Springer},
  \bibinfo{pages}{158--169}.
\newblock
\urldef\tempurl%
\url{https://doi.org/10.1007/978-3-642-17572-5\_13}
\showDOI{\tempurl}


\bibitem[Chawla et~al\mbox{.}(2010)]%
        {Chawla}
\bibfield{author}{\bibinfo{person}{Shuchi Chawla}, \bibinfo{person}{Jason~D.
  Hartline}, \bibinfo{person}{David~L. Malec}, {and}
  \bibinfo{person}{Balasubramanian Sivan}.} \bibinfo{year}{2010}\natexlab{}.
\newblock \showarticletitle{Multi-parameter mechanism design and sequential
  posted pricing}. In \bibinfo{booktitle}{\emph{{BQGT}}}.
  \bibinfo{publisher}{{ACM}}, \bibinfo{pages}{22:1}.
\newblock
\urldef\tempurl%
\url{https://doi.org/10.1145/1807406.1807428}
\showDOI{\tempurl}


\bibitem[Correa et~al\mbox{.}(2017)]%
        {Correa0.632}
\bibfield{author}{\bibinfo{person}{Jos\'{e} Correa}, \bibinfo{person}{Patricio
  Foncea}, \bibinfo{person}{Ruben Hoeksma}, \bibinfo{person}{Tim Oosterwijk},
  {and} \bibinfo{person}{Tjark Vredeveld}.} \bibinfo{year}{2017}\natexlab{}.
\newblock \showarticletitle{Posted Price Mechanisms for a Random Stream of
  Customers} \emph{(\bibinfo{series}{EC '17})}. \bibinfo{publisher}{Association
  for Computing Machinery}, \bibinfo{address}{New York, NY, USA},
  \bibinfo{pages}{169–186}.
\newblock
\showISBNx{9781450345279}
\urldef\tempurl%
\url{https://doi.org/10.1145/3033274.3085137}
\showDOI{\tempurl}


\bibitem[Correa et~al\mbox{.}(2021)]%
        {Correa2}
\bibfield{author}{\bibinfo{person}{Jos\'{e} Correa}, \bibinfo{person}{Patricio
  Foncea}, \bibinfo{person}{Ruben Hoeksma}, \bibinfo{person}{Tim Oosterwijk},
  {and} \bibinfo{person}{Tjark Vredeveld}.} \bibinfo{year}{2021}\natexlab{}.
\newblock \showarticletitle{Posted Price Mechanisms and Optimal Threshold
  Strategies for Random Arrivals}.
\newblock \bibinfo{journal}{\emph{Math. Oper. Res.}} \bibinfo{volume}{46},
  \bibinfo{number}{4} (\bibinfo{date}{nov} \bibinfo{year}{2021}),
  \bibinfo{pages}{1452–1478}.
\newblock
\showISSN{0364-765X}
\urldef\tempurl%
\url{https://doi.org/10.1287/moor.2020.1105}
\showDOI{\tempurl}


\bibitem[Correa et~al\mbox{.}(2019a)]%
        {CorreaPPM}
\bibfield{author}{\bibinfo{person}{Jos{\'{e}}~R. Correa},
  \bibinfo{person}{Patricio Foncea}, \bibinfo{person}{Dana Pizarro}, {and}
  \bibinfo{person}{Victor Verdugo}.} \bibinfo{year}{2019}\natexlab{a}.
\newblock \showarticletitle{From pricing to prophets, and back!}
\newblock \bibinfo{journal}{\emph{Oper. Res. Lett.}} \bibinfo{volume}{47},
  \bibinfo{number}{1} (\bibinfo{year}{2019}), \bibinfo{pages}{25--29}.
\newblock
\urldef\tempurl%
\url{https://doi.org/10.1016/j.orl.2018.11.010}
\showDOI{\tempurl}


\bibitem[Correa et~al\mbox{.}(2019b)]%
        {Correa}
\bibfield{author}{\bibinfo{person}{Jos{\'{e}}~R. Correa},
  \bibinfo{person}{Raimundo Saona}, {and} \bibinfo{person}{Bruno Ziliotto}.}
  \bibinfo{year}{2019}\natexlab{b}.
\newblock \showarticletitle{Prophet Secretary Through Blind Strategies}. In
  \bibinfo{booktitle}{\emph{{SODA}}},
  \bibfield{editor}{\bibinfo{person}{Timothy~M. Chan}} (Ed.).
  \bibinfo{publisher}{{SIAM}}, \bibinfo{pages}{1946--1961}.
\newblock
\urldef\tempurl%
\url{https://doi.org/10.1137/1.9781611975482.118}
\showDOI{\tempurl}


\bibitem[Ehsani et~al\mbox{.}(2018)]%
        {Ehsani}
\bibfield{author}{\bibinfo{person}{Soheil Ehsani},
  \bibinfo{person}{MohammadTaghi Hajiaghayi}, \bibinfo{person}{Thomas
  Kesselheim}, {and} \bibinfo{person}{Sahil Singla}.}
  \bibinfo{year}{2018}\natexlab{}.
\newblock \showarticletitle{Prophet Secretary for Combinatorial Auctions and
  Matroids}. In \bibinfo{booktitle}{\emph{{SODA}}}.
  \bibinfo{publisher}{{SIAM}}, \bibinfo{pages}{700--714}.
\newblock
\urldef\tempurl%
\url{https://doi.org/10.1137/1.9781611975031.46}
\showDOI{\tempurl}


\bibitem[Esfandiari et~al\mbox{.}(2015)]%
        {Esfandiari}
\bibfield{author}{\bibinfo{person}{Hossein Esfandiari},
  \bibinfo{person}{MohammadTaghi Hajiaghayi}, \bibinfo{person}{Vahid Liaghat},
  {and} \bibinfo{person}{Morteza Monemizadeh}.}
  \bibinfo{year}{2015}\natexlab{}.
\newblock \showarticletitle{Prophet Secretary}. In
  \bibinfo{booktitle}{\emph{{ESA}}} \emph{(\bibinfo{series}{Lecture Notes in
  Computer Science}, Vol.~\bibinfo{volume}{9294})}.
  \bibinfo{publisher}{Springer}, \bibinfo{pages}{496--508}.
\newblock
\urldef\tempurl%
\url{https://doi.org/10.1007/978-3-662-48350-3\_42}
\showDOI{\tempurl}


\bibitem[Giambartolomei et~al\mbox{.}(2023)]%
        {Giordano}
\bibfield{author}{\bibinfo{person}{Giordano Giambartolomei},
  \bibinfo{person}{Frederik Mallmann-Trenn}, {and} \bibinfo{person}{Raimundo
  Saona}.} \bibinfo{year}{2023}\natexlab{}.
\newblock \bibinfo{title}{Prophet Inequalities: Separating Random Order from
  Order Selection}.
\newblock
\newblock
\showeprint[arxiv]{2304.04024}~[cs.DS]


\bibitem[Hajiaghayi et~al\mbox{.}(2007)]%
        {PPM}
\bibfield{author}{\bibinfo{person}{Mohammad~Taghi Hajiaghayi},
  \bibinfo{person}{Robert Kleinberg}, {and} \bibinfo{person}{Tuomas Sandholm}.}
  \bibinfo{year}{2007}\natexlab{}.
\newblock \showarticletitle{Automated Online Mechanism Design and Prophet
  Inequalities}. In \bibinfo{booktitle}{\emph{Proceedings of the 22nd National
  Conference on Artificial Intelligence - Volume 1}} (Vancouver, British
  Columbia, Canada) \emph{(\bibinfo{series}{AAAI'07})}.
  \bibinfo{publisher}{AAAI Press}, \bibinfo{pages}{58–65}.
\newblock
\showISBNx{9781577353232}


\bibitem[Hill and Kertz(1982)]%
        {Hill}
\bibfield{author}{\bibinfo{person}{T.~P. Hill} {and} \bibinfo{person}{Robert~P.
  Kertz}.} \bibinfo{year}{1982}\natexlab{}.
\newblock \showarticletitle{{Comparisons of Stop Rule and Supremum Expectations
  of I.I.D. Random Variables}}.
\newblock \bibinfo{journal}{\emph{The Annals of Probability}}
  \bibinfo{volume}{10}, \bibinfo{number}{2} (\bibinfo{year}{1982}),
  \bibinfo{pages}{336 -- 345}.
\newblock
\urldef\tempurl%
\url{https://doi.org/10.1214/aop/1176993861}
\showDOI{\tempurl}


\bibitem[Kleinberg and Weinberg(2012)]%
        {Kleinberg}
\bibfield{author}{\bibinfo{person}{Robert Kleinberg} {and}
  \bibinfo{person}{S.~Matthew Weinberg}.} \bibinfo{year}{2012}\natexlab{}.
\newblock \showarticletitle{Matroid prophet inequalities}. In
  \bibinfo{booktitle}{\emph{{STOC}}}. \bibinfo{publisher}{{ACM}},
  \bibinfo{pages}{123--136}.
\newblock
\urldef\tempurl%
\url{https://doi.org/10.1145/2213977.2213991}
\showDOI{\tempurl}


\bibitem[Krengel and Sucheston(1977)]%
        {Krengel1}
\bibfield{author}{\bibinfo{person}{Ulrich Krengel} {and} \bibinfo{person}{Louis
  Sucheston}.} \bibinfo{year}{1977}\natexlab{}.
\newblock \showarticletitle{Semiamarts and finite values}.
\newblock \bibinfo{journal}{\emph{Bull. Amer. Math. Soc.}}
  \bibinfo{volume}{83}, \bibinfo{number}{4} (\bibinfo{year}{1977}),
  \bibinfo{pages}{745--747}.
\newblock


\bibitem[Krengel and Sucheston(1978)]%
        {Krengel2}
\bibfield{author}{\bibinfo{person}{Ulrich Krengel} {and} \bibinfo{person}{Louis
  Sucheston}.} \bibinfo{year}{1978}\natexlab{}.
\newblock \showarticletitle{On semiamarts, amarts, and processes with finite
  value}.
\newblock \bibinfo{journal}{\emph{Adv. in Probability}}  \bibinfo{volume}{4}
  (\bibinfo{year}{1978}), \bibinfo{pages}{197--266}.
\newblock


\bibitem[Lee and Singla(2018)]%
        {Lee}
\bibfield{author}{\bibinfo{person}{Euiwoong Lee} {and} \bibinfo{person}{Sahil
  Singla}.} \bibinfo{year}{2018}\natexlab{}.
\newblock \showarticletitle{Optimal Online Contention Resolution Schemes via
  Ex-Ante Prophet Inequalities}. In \bibinfo{booktitle}{\emph{{ESA}}}
  \emph{(\bibinfo{series}{LIPIcs}, Vol.~\bibinfo{volume}{112})},
  \bibfield{editor}{\bibinfo{person}{Yossi Azar}, \bibinfo{person}{Hannah
  Bast}, {and} \bibinfo{person}{Grzegorz Herman}} (Eds.).
  \bibinfo{publisher}{Schloss Dagstuhl - Leibniz-Zentrum f{\"{u}}r Informatik},
  \bibinfo{pages}{57:1--57:14}.
\newblock
\urldef\tempurl%
\url{https://doi.org/10.4230/LIPIcs.ESA.2018.57}
\showDOI{\tempurl}


\bibitem[Liu et~al\mbox{.}(2021)]%
        {Liu}
\bibfield{author}{\bibinfo{person}{Allen Liu}, \bibinfo{person}{Renato~Paes
  Leme}, \bibinfo{person}{Martin P{\'{a}}l}, \bibinfo{person}{Jon Schneider},
  {and} \bibinfo{person}{Balasubramanian Sivan}.}
  \bibinfo{year}{2021}\natexlab{}.
\newblock \showarticletitle{Variable Decomposition for Prophet Inequalities and
  Optimal Ordering}. In \bibinfo{booktitle}{\emph{{EC}}}.
  \bibinfo{publisher}{{ACM}}, \bibinfo{pages}{692}.
\newblock
\urldef\tempurl%
\url{https://doi.org/10.1145/3465456.3467598}
\showDOI{\tempurl}


\bibitem[Lucier(2017)]%
        {Lucier}
\bibfield{author}{\bibinfo{person}{Brendan Lucier}.}
  \bibinfo{year}{2017}\natexlab{}.
\newblock \showarticletitle{An economic view of prophet inequalities}.
\newblock \bibinfo{journal}{\emph{SIGecom Exch.}} \bibinfo{volume}{16},
  \bibinfo{number}{1} (\bibinfo{year}{2017}), \bibinfo{pages}{24--47}.
\newblock
\urldef\tempurl%
\url{https://doi.org/10.1145/3144722.3144725}
\showDOI{\tempurl}


\bibitem[Peng and Tang(2022)]%
        {PT}
\bibfield{author}{\bibinfo{person}{Bo Peng} {and} \bibinfo{person}{Zhihao~Gavin
  Tang}.} \bibinfo{year}{2022}\natexlab{}.
\newblock \showarticletitle{Order Selection Prophet Inequality: From Threshold
  Optimization to Arrival Time Design}. In \bibinfo{booktitle}{\emph{{FOCS}}}.
\newblock


\end{thebibliography}
\end{document}